%% file: main.tex
\documentclass[a4paper,UKenglish, autoref, thm-restate,dvipsnames,pdfa,cleveref]{lipics-v2021} %

\usepackage{emily}
\usepackage{xcolor}
\title{Kamp Theorem for Pomset Languages of Higher Dimensional Automata} %

\author{Emily Clement}{CNRS, LIPN UMR 7030, Université Sorbonne Paris Nord, F-93430 Villetaneuse, France}{emily.clement@lipn.univ-paris13.fr}{https://orcid.org/0000-0002-8105-665X}{}
\author{Enzo Erlich}{Université Paris Cité, CNRS, IRIF, F-75013, Paris, France  \and EPITA Research Laboratory (LRE), Paris, France}{erlich@irif.fr}{https://orcid.org/0000-0001-7375-4725}{}%
\author{Jérémy Ledent}{Université Paris Cité, CNRS, IRIF, F-75013, Paris, France}{jeremy.ledent@irif.fr}{https://orcid.org/0009-0007-3574-6362}{}%

\authorrunning{E. Clement, E. Erlich and J. Ledent} %

\Copyright{Emily Clement, Enzo Erlich and Jérémy Ledent} %



\relatedversion{} %

\funding{This work was partly supported by ANR projects BisoUS (ANR-22-CE48-0012) and PaVeDyS (ANR-23-CE48-0005).}%

\nolinenumbers %

\hideLIPIcs

\begin{document}

\maketitle

\begin{abstract}
\input{sub/abstract}
\end{abstract}

\newpage

\section{Introduction}
\label{sec-intro}

\input{sub/intro}

\section{Preliminaries}
\label{sec-defs}
\input{sub/defs}

\section{First Order logic over pomsets and ST-sequences}
\label{sec-FO}
\input{sub/fo-pomset-to-fo-st-seq.tex}

\section{A Linear Temporal Logic for Pomsets}
\label{sec-ltl}
\input{sub/ltl-pomsets}

\section{Expressivity}
\label{sec-ltl-expressivity}
\input{sub/expressivity}

\section{Conclusion and future work}
\label{sec-conclu}
\input{sub/conclusion}

\bibliography{biblio}

\clearpage
\appendix

\section{Discussion: Why we need sub-events}
\label{app:sub-events}

\input{sub/discussion.tex}

\section{Omitted proofs}
\label{app:proofs}

\subsection{Proof of \Cref{lem-sim-FO}}
\label{app:proof-lem-sim-FO}
\input{sub/appendix-sim-FO}

\subsection{Proof of \Cref{thm-FOP-to-FOST}}
\label{app:proof-thm-FOP-to-FOST}
\input{sub/appendix-FOP} %

\subsection{Proof of \Cref{prop-ltlp-sube}}
\label{app:proof-prop-ltlp-sube}
\input{sub/appendix-ltlp-sube}

\subsection{Proof of \Cref{prop-induction}}
\label{app:proof-prop-induction}
\input{sub/appendix-ltlp-induction} %

\subsection{Proof of \Cref{lem-LTLP-to-FOP}}
\label{app:proof-lem-ltlp-to-fop}
\input{sub/appendix-ltlp-to-fop} %

\end{document}

%% file: sub/abstract.tex
 Temporal logics are a powerful tool to specify properties of computational systems.
 For concurrent programs, Higher Dimensional Automata (HDA) are a very expressive model of non-interleaving concurrency.
 HDA recognize languages of partially ordered multisets, or pomsets.
 Recent work has shown that Monadic Second Order (MSO) logic is as expressive as HDA for pomset languages.
 In the case of words, Kamp's theorem states that First Order (FO) logic is as expressive as Linear Temporal Logic (LTL).
 In this paper, we extend this result to pomsets.
 To do so, we first investigate the class of pomset languages that are definable in FO.
 As expected, this is a strict subclass of MSO-definable languages.
 Then, we define a Linear Temporal Logic for pomsets (\lsptl), and show that it is equivalent to FO.

%% file: sub/intro.tex
\subparagraph*{Context.}

Higher-Dimensional Automata (HDA)~\cite{Pratt91} are a powerful model of concurrency that enriches standard finite-state automata with higher-dimensional cells, allowing to specify that some events may occur simultaneously.
HDA represent asynchronous concurrent computations where events are non-atomic, and can overlap in various complex patterns. 
Thus, they are a model of \emph{non-interleaving} concurrency, which is concerned with simultaneity between events rather than commutation of events.
HDA have been shown to be a very expressive model of concurrency: comparisons with other models (such as Petri nets, event structures) can be found in~\cite{Glabbeek06}, based on history-preserving bisimulations, or in~\cite{GoubaultM12}, based on adjunctions.

\begin{figure}[htbp]
	\centering
	\begin{tikzpicture}[scale=1, >=stealth, shorten >=1pt, shorten <=1pt, font=\footnotesize]
		\foreach \x in {0,1,2}
		{
			\foreach \y in {0,1}
			{
				\draw[->, semithick] (\x,\y+0.1) -- (\x,\y+0.9) ;
			}
		}
		\foreach \x in {0,1,2}
		{
			\foreach \y in {0,1}
			{
				\draw[->, semithick] (\y+0.1,\x) -- (\y+0.9,\x) ;
			}
		}
		\node at (0.5,-0.2) {$a$};
		\node at (1.5,-0.2) {$c$};
		\node at (-0.2,0.5) {$b$};
		\node at (-0.2,1.5) {$d$};
		\foreach \x / \y in { 0/0 , 0/1 , 1/0 , 1/1 }
		{
			\fill[fill opacity = 0.1] (\x,\y) -- (\x+1,\y) --(\x+1,\y+1) --(\x,\y+1) -- cycle  ;
		}
		\draw[->, thick, color=red] (0,0) .. controls (0.5,1.5) .. (1.9,1.9);
		\foreach \x in {0,1,2}
		{\foreach \y in {0,1,2}
			{\node[state, fill=white, inner sep=0pt, minimum size=8pt] (v\x\y) at (\x,\y) {};
		}}
		\node[state, initial, initial text={}, fill=white, inner sep=0pt, minimum size=8pt] (v00) at (0,0) {};
		\node[state, accepting, fill=white, inner sep=0pt, minimum size=8pt] (v22) at (2,2) {};
	\end{tikzpicture}
	\qquad \qquad
	\begin{tikzpicture}[scale=1, baseline=-0.8cm]
		\def\hw{0.3}
		\def\fh{0.7}
		\def\sh{0.2}
		\coordinate (O) at (0,0);
		\coordinate[above = 1cm of O] (Up);
		\coordinate[right = 3cm of O] (Right);
		\coordinate[above = 1cm of Right] (Up-right);
		\path[->, >=latex] (0,0) edge node[font=\scriptsize, below, pos=0.9] {Time} (3.2,0);

		\filldraw[pomset-1] (0.2,\fh) -- (1.8,\fh) -- (1.8,\fh+\hw) -- (0.2, \fh+\hw) --  cycle; 
		\node at (1,\fh+\hw*0.5) {$a$};
		
		\filldraw[pomset-4] (2,\fh) -- (3,\fh) -- (3,\fh+\hw) -- (2, \fh+\hw) --  cycle; 
		\node at (2.5,\fh+\hw*0.5) {$c$};
		
		\filldraw[pomset-2] (0.2,\sh) -- (1.1,\sh) -- (1.1,\sh+\hw) -- (0.2, \sh+\hw) --  cycle; 
		\node at (0.65,\sh+\hw*0.5) {$b$};
		
		\filldraw[pomset-3] (1.3,\sh) -- (3,\sh) -- (3,\sh+\hw) -- (1.3, \sh+\hw) --  cycle; 
		\node at (2.15,\sh+\hw*0.5) {$d$};
	\end{tikzpicture}
	\qquad \qquad
	\begin{tikzpicture}[scale=1, baseline=-1.3cm]
		\coordinate (A) at (0,0.5);
		\node (a) at (A) {$a$};
		\node[right of = a, xshift = 0.2cm] (c) {$c$};
		\node[below of = a, yshift=0.05cm] (b) {$b$};
		\node[right of = b, xshift = 0.2cm] (d) {$d$};
		\path[precedence] (b) edge (d) (a) edge (c) (b) edge (c);
	\end{tikzpicture}
	\caption{An HDA execution, depicted with intervals, and as a pomset.}
	\label{fig:HDA-pomset}
\end{figure}

A notion of language of HDA was developed only recently~\cite{FahrenbergJSZ21}.
The key idea is that an execution of an HDA may not be represented as a word (where the events/letters are totally ordered); instead, it is a \emph{partially ordered multiset}, or \emph{pomset} for short.
This is illustrated in \Cref{fig:HDA-pomset}. On the left, an HDA is depicted, with two processes running in parallel. One process is executing two events $a$, then~$c$, while the second process is running $b$, then~$d$.
One possible execution path is depicted in red. %
This execution can equivalently be depicted using intervals (center of \Cref{fig:HDA-pomset}): time flows from left to right, and the duration of events is depicted by stretching horizontally the boxes representing the events. When two events overlap vertically, this means that they are occurring simultaneously.
Finally, a third representation of this execution is depicted on the right of \Cref{fig:HDA-pomset}, as a pomset.
The events are ordered according to Lamport's \emph{happens before} relation~\cite{Lamport86}: $a \to c$ means that event $a$ terminates before $c$ starts. In particular, overlapping events are not ordered.

Thus, HDA recognize languages of pomsets.
Note that, just as sequential executions are a special case of concurrent executions; words are also a special case of pomsets, where the order between events happens to be total.
A recent line of work has been extending classic results of automata theory to the setting of pomset languages of HDA.
They have been shown to enjoy a variant of Kleene theorem~\cite{FahrenbergJSZ22}, a Myhill-Nerode theorem~\cite{FahrenbergZ23}, a pumping Lemma~\cite{AmraneBFZ23}, and a B\"uchi-Elgot-Trakhtenbrot theorem~\cite{AmraneBFF24}.
The latter result will be of particular importance for us: it says that Monadic Second-Order (MSO) logic describes the same class of languages as HDA.
To establish this result, the authors used an equivalent representation of pomsets called \emph{ST-sequences} (see \Cref{sec:ST-decomposition}),
which are words over a finite alphabet.
This proof technique, further developed in~\cite{AmraneBCFZ24}, provides a way to lift some results from words to pomsets.

In this paper, we investigate the class of languages obtained by restricting to First Order~(FO) logic instead of MSO.
Over finite words, FO-definable languages form a strict subclass of regular languages, called \emph{aperiodic languages}~\cite{Schutzenberger65}.
There are many equivalent characterizations of this class, such as star-free regular expressions and counter-free automata~\cite{McNaughton71}.
Of particular interest to us is Kamp's theorem~\cite{Kamp68, Rabinovich14}, which states that Linear Temporal Logic (LTL) is as expressive as FO logic over finite words.

\subparagraph*{Contributions: FO and LTL for pomset languages.} 
This paper studies the class of FO-definable languages of pomsets. As expected, they form a strict subclass of HDA languages (\Cref{cor-FO<MSO}).
Our main goal is to define an LTL-style logic over pomsets, that defines the same class of languages as FO logic; thus extending Kamp's theorem to pomset languages.
To that end, we first show that FO formulas over pomsets can be translated into equivalent FO formulas over ST-sequences (\Cref{thm-FOP-to-FOST}).
A similar result for MSO formulas was proved in~\cite{AmraneBFF24}. However, it cannot be easily adapted: it relies heavily on a relation~$\relation$, that relates different occurrences of the same event in an ST-sequence. 
To show that this relation is definable in MSO, the authors use second-order quantification.
In our first technical contribution (\Cref{lem-sim-FO}), we show that the relation~$\relation$ is FO-definable.
Instead of giving directly an FO formula (which would be very tedious), we define a counter-free automaton.

We then define a variant of LTL for pomsets that we call \lsptl 
and prove that \lsptl is equivalent to \fo (\Cref{thm-LTLP--FOP}).
We sum up the proof of \Cref{thm-sptl=fo} by the chain of translations depicted in \Cref{fig-translations}.
\begin{figure}
    \centering
    \input{figures/translations.tex}
	\vspace{-12pt}
    \caption{Summary of the proof of \Cref{thm-sptl=fo}.}
    \label{fig-translations}
\end{figure}

\subparagraph*{Relationship with Mazurkiewicz traces.}
Pomsets also appear in the literature in the context of Mazurkiewicz traces~\cite{Mazurkiewicz86Trace}, a notion also at the intersection between automata theory and concurrency theory.
Many variants of LTL on traces have been defined, and their expressivity has been extensively studied~\cite{Thiagarajan94traceLTL, DiekertG02completeLTL, ThiagarajanW02completeLTL, DiekertG06local, DiekertG06local-global, GastinK07PSPACE}.
Thus, it is important to distinguish this line of work with what we are studying in this paper.

The set of Mazurkiewicz traces can be described as the free partially commutative monoid over a dependence alphabet. Thus, a trace is an equivalence class of finite words where some letters are allowed to commute.
Equivalently, they can also be regarded as pomsets, using their \emph{dependence graph}.
In the dependence graph of a Mazurkiewicz trace, $a \to b$ denotes causal dependency. Two incomparable events are independent, i.e., it does not matter in which order they occur.
In contrast, for HDA pomsets, the precedence order $a \to b$ indicates that~$a$ terminates before $b$ starts, and two events are incomparable when they are simultaneous (i.e., the time intervals during which they are executed overlap). 
This semantic distinction results in several technical differences:
\begin{enumerate}[(i)]
\item HDA pomsets are always required to be \emph{interval orders} (see \Cref{def-interval-order}). This is not the case for Mazurkiewicz traces.
\item For HDA, it may very well be the case that $a \to b$ (i.e., $a$ happens before~$b$) in one execution, while $a$ and $b$ are incomparable (i.e., simultaneous) in another execution.
This is not allowed for Mazurkiewicz traces. Indeed, they are defined with respect to a dependence alphabet, which determines in advance which events may or may not be comparable in the dependence graph.
\item Lastly, HDA pomsets actually have an extra relation $\eventorder$ called the \emph{event order}.
Indeed, we require simultaneous events (events that are incomparable for the relation $\to$) to be ordered by the event order.
This is a way of managing process identity: if two processes are running~$a$ in parallel, the event order allows to distinguish the two events~$a$.
\end{enumerate}

\subparagraph*{Other related work.}
Modal logics over HDA have been previously investigated.
In~\cite{Prisacariu10}, the author introduces a logic called Higher-Dimensional Modal Logic (HDML), which is interpreted directly on an HDA.
It has two variants of the \enquote*{next} operator: one that can start an event, and one that can terminate an event.
A sound and complete axiomatization of this logic is given, as well as tentative definitions of LTL-like and CTL-like extensions.
The language theory over pomsets is not investigated since pomset languages of HDA had not been defined at the time.
However, an encoding of LTL for Mazurkiewicz traces into HDML is given.
A more recent paper~\cite{ZouariZF24} introduces IPomset Modal Logic (IPML).
This logic features a forward modality and a backwards modality, in the spirit of path logic~\cite{JoyalNW96}.
No temporal variant of IPML is defined, the focus of this paper being on bisimulation equivalence.

\subparagraph*{Plan of the paper.}

In \Cref{sec-defs}, we recall the key definitions that we will be using. %
Then, in \Cref{sec-FO}, we investigate FO over pomsets, and show our first technical contribution: FO over pomsets can be translated to FO over ST-sequences (\Cref{thm-FOP-to-FOST}).
In \Cref{sec-ltl}, we define our temporal logic \lsptl,
and study its expressivity in \Cref{sec-ltl-expressivity} (\Cref{thm-sptl=fo}).

%% file: figures/translations.tex
\begin{tikzpicture}[scale=.8]

    \node[draw, rectangle, rounded corners=.2cm, minimum height=.6cm] (FOP) {\fo for pomsets};
    \node[draw, rectangle, rounded corners=.2cm, minimum height=.6cm] (FOW) [right=5cm of FOP] {\fo for ST-sequences};
    \node (Inter) [below=1cm of FOW] {};
    \node[draw, rectangle, rounded corners=.2cm, minimum height=.6cm] (LTLW) [below=0.8cm of FOW] {\ltl for ST-sequences};
    \node[draw, rectangle, rounded corners=.2cm, minimum height=.6cm] (SPTL) [below=0.8cm of FOP] {\ltl for pomsets (\lsptl)};

    \path[-latex, thick] (FOP) edge node [above] {\cref{thm-FOP-to-FOST}} (FOW);
    \path[latex-latex, thick] (FOW) edge node [right] {Kamp's theorem \cite{Rabinovich14}} (LTLW);
    \path[-latex, thick] (LTLW) edge %
    node[below] {\cref{lem-LTL-to-SPTL}}(SPTL);
    \path[-latex, thick] (SPTL) edge node [left] {\cref{lem-SPTL-to-FOP}} (FOP);
\end{tikzpicture}

%% file: sub/defs.tex
In this section, we define several notions that we will need in this paper: interval pomsets with interfaces~\cite{FahrenbergJSZ21}, their ST-decomposition~\cite{AmraneBCFZ24}, and MSO logic on pomsets~\cite{AmraneBFF24}.
\subsection{Interval pomsets with interfaces}

\input{sub/sub-defs/iiPomsets.tex}

\subsection{ST decomposition}
\label{sec:ST-decomposition}
 \input{sub/sub-defs/ST-decomp.tex}

\subsection{Monadic Second Order logic over pomsets}
\label{sec:MSO-pomsets}
\input{sub/sub-defs/MSO.tex}

%% file: sub/sub-defs/iiPomsets.tex
\subparagraph*{Pomsets, interfaces, dimension.}

Let $ \alphabet $ be a finite alphabet.
A partially ordered multiset (or pomset) over $\alphabet$ is a generalization of words, where the letters need not be totally ordered.
This can represent situations where two or more events occur at the same time.
An intuitive representation of some pomsets is given in \Cref{fig-ex-pomsets}, where the events $a, b, c, d \in \alphabet$ are depicted as intervals, and the horizontal axis represents the elapsing time.
For instance, in \Cref{fig-pomset-1}, both events $a$ and $b$ occur before $c$, which we denote by $a \precedence c$ and $b \precedence c$.
However, since the intervals for $a$ and $b$ overlap, the two events are concurrent: they are incomparable for the precedence relation~$\precedence$.  
\begin{figure}[htbp]
    \centering 
    \begin{subfigure}[t]{0.2\linewidth}
        \centering
		\scalebox{0.8}{
         \begin{tikzpicture}
        	\coordinate (A) at (0,0);
        	\node (a) at (A) {$a$};
        	\node[right = .8cm of a] (c) {$c$};
        	\node[below = .5cm of a] (b) {$b$};
        	\path[precedence] (a) edge (c) (b) edge (c);
        	\path[event-order] (a) edge (b);
        \end{tikzpicture}
		}
    	\medskip

		\scalebox{0.8}{
        \begin{tikzpicture}%
            \def\hw{0.3}
            \def\fh{0.6}
            \def\sh{0.1}
            \coordinate (O) at (0,0);
            \coordinate[above = 1cm of O] (Up);
            \coordinate[right = 3cm of O] (Right);
            \coordinate[above = 1cm of Right] (Up-right);
            \draw[-] (O) -- (Up);
            \draw[-] (Right) -- (Up-right);
        
            \filldraw[pomset-1] (0.2,\fh) -- (1.7,\fh) -- (1.7,\fh+\hw) -- (0.2, \fh+\hw) --  cycle; 
            \node at (0.2+0.75,\fh+\hw*0.5) {$a$};
        
            \filldraw[pomset-4] (2,\fh) -- (2.8,\fh) -- (2.8,\fh+\hw) -- (2, \fh+\hw) --  cycle; 
            \node at (2.4,\fh+\hw*0.5) {$c$};
        
            \filldraw[pomset-2] (0.4,\sh) -- (1.3,\sh) -- (1.3,\sh+\hw) -- (0.4, \sh+\hw) --  cycle; 
            \node at (0.85,\sh+\hw*0.5) {$b$};
        \end{tikzpicture}
		}
    \caption{}%
    \label{fig-pomset-1}
    \end{subfigure}
    \hfill
    \begin{subfigure}[t]{0.2\linewidth}
    \centering
	\scalebox{0.8}{
    \begin{tikzpicture}
    	\coordinate (A) at (0,0);
    	\node (o) at (A) {};
    	\node[right = 0.4cm of A] (a) {$a$};
    	\node[below = 0.5cm of o] (b) {$b$};
    	\node[right = 0.8cm of b] (d) {$d$};
    	\path[precedence] (b) edge (d);
    	\path[event-order] (a) edge (b) (a) edge (d);
    \end{tikzpicture}
	}
	\medskip
	
	\scalebox{0.8}{
    \begin{tikzpicture}%
        \def\hw{0.3}
        \def\fh{0.6}
		\def\sh{0.1}
        \coordinate (O) at (0,0);
        \coordinate[above = 1cm of O] (Up);
        \coordinate[right = 3cm of O] (Right);
        \coordinate[above = 1cm of Right] (Up-right);
        \draw[-] (O) -- (Up);
        \draw[-] (Right) -- (Up-right);

        \filldraw[pomset-1] (0.2,\fh) -- (2.8,\fh) -- (2.8,\fh+\hw) -- (0.2, \fh+\hw) --  cycle; 
        \node at (1.5,\fh+\hw*0.5) {$a$};

        \filldraw[pomset-2] (0.2,\sh) -- (1.1,\sh) -- (1.1,\sh+\hw) -- (0.2, \sh+\hw) --  cycle; 
        \node at (0.65,\sh+\hw*0.5) {$b$};

        \filldraw[pomset-3] (1.4,\sh) -- (2.8,\sh) -- (2.8,\sh+\hw) -- (1.4, \sh+\hw) --  cycle; 
        \node at (2.1,\sh+\hw*0.5) {$d$};
    \end{tikzpicture}
	}
    \caption{}%
    \label{fig-pomset-2}
\end{subfigure}
\hfill
\begin{subfigure}[t]{0.2\linewidth}
    \centering
	\scalebox{0.8}{
    \begin{tikzpicture}
    	\coordinate (A) at (0,0.5);
    	\node (a) at (A) {$a$};
    	\node[right of = a, xshift = 0.2cm] (c) {$c$};
    	\node[below of = a, yshift=0.05cm] (b) {$b$};
    	\node[right of = b, xshift = 0.2cm] (d) {$d$};
    	\path[precedence] (b) edge (d) (a) edge (c) (b) edge (c);
    	\path[event-order] (a) edge (b) (a) edge (d) (c) edge (d);
    \end{tikzpicture}
	}
	\medskip
	
	\scalebox{0.8}{
    \begin{tikzpicture}%
        \def\hw{0.3}
        \def\fh{0.6}
        \def\sh{0.1}
        \coordinate (O) at (0,0);
        \coordinate[above = 1cm of O] (Up);
        \coordinate[right = 3cm of O] (Right);
        \coordinate[above = 1cm of Right] (Up-right);
        \draw[-] (O) -- (Up);
        \draw[-] (Right) -- (Up-right);

        \filldraw[pomset-1] (0.2,\fh) -- (1.7,\fh) -- (1.7,\fh+\hw) -- (0.2, \fh+\hw) --  cycle; 
        \node at (0.2+0.75,\fh+\hw*0.5) {$a$};

        \filldraw[pomset-4] (2,\fh) -- (2.8,\fh) -- (2.8,\fh+\hw) -- (2, \fh+\hw) --  cycle; 
        \node at (2.4,\fh+\hw*0.5) {$c$};

        \filldraw[pomset-2] (0.2,\sh) -- (1.1,\sh) -- (1.1,\sh+\hw) -- (0.2, \sh+\hw) --  cycle; 
        \node at (0.65,\sh+\hw*0.5) {$b$};

        \filldraw[pomset-3] (1.4,\sh) -- (2.8,\sh) -- (2.8,\sh+\hw) -- (1.4, \sh+\hw) --  cycle; 
        \node at (2.1,\sh+\hw*0.5) {$d$};
    \end{tikzpicture}
	}
    \caption{}%
    \label{fig-pomset-3}
\end{subfigure}
\hfill
\begin{subfigure}[t]{0.2\linewidth}
    \centering
	\scalebox{0.8}{
    \begin{tikzpicture}
    	\coordinate (A) at (0,0.5);
    	\node (a) at (A) {$\vphantom{\ibullet}a$};
    	\node[right of = a, xshift = 0.2cm] (c) {$c\vphantom{\ibullet}$};
    	\node[below of = a, yshift = 0.05cm] (b) {$\ibullet b$};
    	\node[right of = b, xshift = 0.2cm] (d) {$d\ibullet$};
    	\path[precedence] (b) edge (d) (a) edge (c) (b) edge (c);
    	\path[event-order] (a) edge (b) (a) edge (d) (c) edge (d);
    \end{tikzpicture}
	}
	\medskip
	
	\scalebox{0.8}{
    \begin{tikzpicture}%
        \def\hw{0.3}
        \def\fh{0.6}
        \def\sh{0.1}
        \coordinate (O) at (0,0);
        \coordinate[above = 1cm of O] (Up);
        \coordinate[right = 3cm of O] (Right);
        \coordinate[above = 1cm of Right] (Up-right);
        \draw[-] (O) -- (0,\sh);
        \draw[-] (0, \sh+\hw) -- (Up);
        \draw[-] (Right) -- (Up-right);

        \filldraw[pomset-1] (0.2,\fh) -- (1.7,\fh) -- (1.7,\fh+\hw) -- (0.2, \fh+\hw) --  cycle; 
        \node at (0.2+0.75,\fh+\hw*0.5) {$a$};

        \filldraw[pomset-4] (2,\fh) -- (2.8,\fh) -- (2.8,\fh+\hw) -- (2, \fh+\hw) --  cycle; 
        \node at (2.4,\fh+\hw*0.5) {$c$};

        \fill[fill=pomset-2-light] (-0.01,\sh) -- (1.1,\sh) -- (1.1,\sh+\hw) -- (-0.01, \sh+\hw) --  cycle; 
        \draw[-] (0,\sh) --(1.1, \sh) --  (1.1,\sh+\hw) -- (0,\sh+\hw);
        \node at (0.55,\sh+\hw*0.5) {$b$};

        \fill[fill=pomset-3-light] (1.4,\sh) -- (3.01,\sh) -- (3.01,\sh+\hw) -- (1.4, \sh+\hw) --  cycle; 
        \draw[-]  (3.01,\sh+\hw) -- (1.4, \sh+\hw) -- (1.4,\sh) -- (3.01,\sh)  ;
        \node at (2.2,\sh+\hw*0.5) {$d$};
    \end{tikzpicture}
	}
    \caption{}%
    \label{fig-pomset-4}
\end{subfigure}
    \caption{Four interval pomsets of dimension 2. Pomset (d) has an interface.}
    \label{fig-ex-pomsets}
\end{figure}

\noindent
As seen in \Cref{fig-ex-pomsets}, the partial orders that we are interested in arise from an interval representation, where $x \precedence y$ means that $x$'s interval terminates before $y$'s interval starts. This is called an interval order.

\begin{definition}[Interval order]\label{def-interval-order}
    Let $(\pom, \precedence)$ be a partially ordered set.
    The relation $\precedence$ is an \defs{interval order} if there exist $ \lowInt, \upInt : \pom \rightarrow \bbR $, with $\lowInt(x) \leq \upInt(x)$,  satisfying the following condition:
    $\forall x, y \in \pom.\quad x \precedence y \iff \upInt(x) < \lowInt(y)$
\end{definition}

\noindent
Not all partial orders are interval orders: for example, the $(2+2)$ pomset depicted in \Cref{fig-noninterval} does not have an interval representation.
Indeed, if we try to assign intervals to the four events, with $a$ before $c$ and $b$ before $d$, we always end up with an extra relation: either $a$ before $d$, or $b$ before $c$, or both.
In fact, interval orders can be characterized as exactly those partial orders that do not have an induced subposet isomorphic to $(2+2)$.

\begin{figure}[h]
    \centering
        \begin{tikzpicture}
            \coordinate (A) at (0,0.5);
            \node (a) at (A) {$a$};
            \node[right of = a, xshift = 0.2cm] (c) {$c$};
            \node[below of = a, yshift = 0.3cm] (b) {$b$};
            \node[right of = b, xshift = 0.2cm] (d) {$d$};
            \path[precedence] (b) edge (d) (a) edge (c);
        \end{tikzpicture}
    \caption{Example of a non-interval pomset: the $2+2$ partial order.}
    \label{fig-noninterval}
\end{figure}

We can now define the notion of pomsets, interval pomsets, and their variants with interfaces.
Pomsets are also equipped with a binary relation $\eventorder$ called the event order, which orders concurrent events.
The intuition of pomsets with interfaces is that we allow some events to be already active at the beginning of a pomset, or still running at the end (see \Cref{fig-pomset-4}).
\begin{definition}[iiPomset]\label{def-pomset}
A \defs{partially ordered multiset} (also called \defs{pomset}) over an alphabet $ \alphabet $ 
is a tuple $( \pom, \precedence_{\pom}, \eventorder_{\pom}, \labelling_{\pom}) $ where 
$ \pom $ is a finite set, $ \precedence_{\pom} $ is a strict partial order over $\pom$ called \defs{precedence}, and $\eventorder_{\pom}$ is an acyclic relation on $\pom$ called the \defs{event order}, and $ \labelling_{\pom} \colon \pom \to \alphabet$ is a labeling function,
 s.t.\  for all $x, y \in \pom$, exactly one of the following holds: 
 \[x = y,\quad x \precedence_{\pom} y,\quad y \precedence_{\pom} x,\quad x \eventorder_{\pom} y,\quad y \eventorder_{\pom} x.\]
 We write $x \parallel_{\pom} y$ when $x \not<_{\pom} y$ and $y \not<_P x$.
 When there is no ambiguity, 
 we denote $ \precedence_{\pom}, \parallel_{\pom}, \eventorder_{\pom} $ and $ \labelling_{\pom} $ as 
 $ \precedence, \parallel, \eventorder$ and $ \labelling $.
\begin{itemize}
\item A pomset $( \pom, \precedence_{\pom}, \eventorder_{\pom}, \labelling_{\pom}) $ is an \defs{interval pomset} if 
$(\pom, \precedence_{\pom})$ is an interval order.
\item An \defs{interval pomset with interfaces} is an interval pomset $(\pom, \precedence_{\pom}, \eventorder_{\pom}, \labelling_{\pom})$ together with two sets $\sint{\pom} \subseteq \pom$ and $\tint{\pom} \subseteq \pom$, called the starting (\resp terminating) interfaces. We require elements of $\sint{\pom}$ (\resp $\tint{\pom}$) to be minimal (\resp maximal) elements w.r.t.~$\precedence_{\pom}$.
\end{itemize}
\end{definition}

\noindent
The set of interval pomsets with interfaces is denoted $\iipoms$.
Note that pomsets are a special case of pomsets with interfaces, where the interfaces are $\sint{P} = \tint{P} = \emptyset$.
In the rest of the paper, pomsets are always assumed to be interval pomsets with interfaces, so we drop the extra adjectives.
Moreover, pomsets are (often implicitly) considered up to isomorphism: the underlying set~$P$ itself does not matter as long as the rest of the structure is the same.

The \defs{dimension} of a pomset $\pom$
is the size of a maximal $\precedence$-antichain in~$\pom$, that is, a maximal set of elements of $\pom$ that are pairwise incomparable w.r.t.\ the precedence relation~$\precedence$.
Such events are called \defs{concurrent}. Note that any set of concurrent events is totally ordered by the event order $\eventorder$.
Intuitively, the dimension of a pomset is the maximal number of processes running concurrently at any time during this execution. 
We denote by $\iipomsk{k}$ the set of pomsets of dimension $\leq k$.

When we draw pomsets in pictures, we use a plain arrow
\begin{tikzpicture}[baseline=-0.5ex]
	\draw[precedence] (0,0) -- (0.5, 0);
\end{tikzpicture} instead of $\precedence$ for the precedence order, 
and a gray dashed arrow
\begin{tikzpicture}[baseline=-0.5ex]
    \draw[event-order] (0,0) -- (0.5, 0);
\end{tikzpicture} for the event order.
We represent the interfaces using a bullet symbol $\ibullet$. We denote an event~$a$ as $\ibullet a$ when it belongs to the starting interface~$\sint{\pom}$; $a \ibullet$ when it belongs to the terminating interface~$\tint{\pom}$; and $\ibullet a \ibullet$ when it belongs to both.
Four pomsets (with and without interface) are depicted in \Cref{fig-ex-pomsets}, next to their interval representation.

\subparagraph*{Gluing of pomsets.}
Gluing is an operation on pomsets that extends word concatenation.
The gluing of two pomsets $P$ and $Q$ is a pomset $P \glue Q$, where all the events of $P$ happen before those of $Q$.
However, we must also take care of the interfaces of~$P$ and~$Q$.
Indeed, if an event is still active when finishing $P$, it must be active when beginning $Q$.
Hence, the terminating interface of $P$ and the starting interface of $Q$ must match.
Formally, we can view $\tint{P}$ and $\sint{Q}$ as pomsets (inheriting the labeling and event order from $P$ and $Q$, respectively), and we require that $\tint{P}$ and $\sint{Q}$ must be isomorphic as pomsets (\emph{i.e.}, the labels and event order must be preserved).
For the purpose of this paper, we only define gluing when the interfaces match exactly ($T_P = S_Q$ as pomsets); see~\cite{FahrenbergJSZ22} for a more robust definition.
For instance, the pomset of \Cref{fig-pomset-4} can be obtained as: 

\begin{center}
\begin{tikzpicture}[scale=0.8, baseline=-0.5ex, every node/.style={transform shape}]
    \coordinate (A) at (0,-0.5);
    \node (b) at (A) {$ \ibullet b $};
    \node[right of = b] (d) {$\vphantom{\ibullet}  d \ibullet$};
    \node[above of = d] (a) {$\vphantom{\ibullet}a \ibullet$};
    \path[precedence] (b) edge (d);
    \path[event-order] (a) edge (b) (a) edge (d);
\end{tikzpicture} \;$\glue$\; 
\begin{tikzpicture}[scale=0.8, baseline=-0.5ex, every node/.style={transform shape}]
    \coordinate (A) at (0,0.5);
    \node (a) at (A) {$ \ibullet a \vphantom{\ibullet}$};
    \node[right of = a] (c) {$ \vphantom{\ibullet}c\vphantom{\ibullet}$};
    \node[below of = a] (d) {$ \ibullet d\ibullet$};
    \path[precedence] (a) edge (c);
    \path[event-order] (a) edge (d) (d) edge (c);
\end{tikzpicture} \;=\;
\begin{tikzpicture}[scale=0.8, baseline=-0.5ex, every node/.style={transform shape}]
    \coordinate (A) at (0,0.5);
    \node (a) at (A) {$\vphantom{\ibullet}a$};
    \node[right of = a, xshift = 0.2cm] (c) {$c\vphantom{\ibullet}$};
    \node[below of = a, yshift = 0.05cm] (b) {$\ibullet b$};
    \node[right of = b, xshift = 0.2cm] (d) {$d\ibullet$};
    \path[precedence] (b) edge (d) (a) edge (c) (b) edge (c);
    \path[event-order] (a) edge (b) (a) edge (d) (c) edge (d);
\end{tikzpicture}
\end{center}

\begin{definition}[Gluing] \label{def-gluing}
    Let $P,Q \in \iipoms$ be two pomsets such that $P \cap Q = \tint{P} = \sint{Q}$.
    The gluing of $P$ and $Q$, denoted by $P \glue Q$, is defined as
    $ P \glue Q := ( R, \precedence_{R}, \eventorder_{R}, \sint{R}, \tint{R}, \labelling_{R} ) $ where:
    \begin{align*}
    R &\egual P \cup Q & \labelling_{R} &\egual \labelling_{P} \cup \labelling_{Q} \\
    \precedence_{R} &\egual \precedence_{P} \cup \precedence_{Q} \cup\; ( P \setminus \tint{P}) \times ( Q \setminus \sint{Q}) & \sint{R} &\egual \sint{P}\\
    \eventorder_{R} &\egual \eventorder_{P} \cup \eventorder_{Q} & \tint{R} &\egual \tint{Q}
    \end{align*}
\end{definition}

%% file: sub/sub-defs/ST-decomp.tex
Starter-Terminator decomposition of pomsets is a tool introduced in \cite{AmraneBFZ23,FahrenbergZ23} to decompose 
a pomset as a gluing of elementary elements, \ie pomsets with empty precedence order.
This technique allows to describe pomsets over $\alphabet$ of dimension $k$ as finite words over a finite alphabet $\stset{k}$.
This makes it possible to lift results from words to pomsets.

A pomset $\pom$ is called \defs{discrete} when it has an empty precedence order.
In that case, the event order $\eventorder$ is a total order.
Thus, we will write discrete pomsets as lists of events, between square brackets, where the event order is omitted and goes implicitly from top to bottom.
For instance, $\pom = \pomset{\pinterface a \ibullet}{\ibullet b \pinterface}$ is a discrete pomset with two concurrent events $a \eventorder b$, where $\sint{\pom} = \{b\}$ and $\tint{\pom} = \{a\}$. 
A discrete pomset $\pom$ can be:
\begin{itemize}
\item a \defs{conclist} if $\sint{\pom} = \tint{\pom} = \emptyset$, 
\item a \defs{starter} if $\tint{\pom} = \pom$,
\item a \defs{terminator} if $\sint{\pom} = \pom$,
\item an \defs{identity} if it is both a starter and a terminator.
\end{itemize}

For example, $\pomset{a}{b}$ is a conclist, $\pomset{a \ibullet}{b \ibullet}$ is a starter, $\pomset{\ibullet a \pinterface}{\ibullet b \ibullet}$ is a terminator, and $\pomset{\ibullet a \ibullet}{\ibullet b \ibullet}$ is an identity.
Intuitively, a starter can only start new events: so, all events must belong to the terminating interface, because they must keep running,
hence the $\tint{\pom} = \pom$. Conversely, a terminator is allowed to terminate some events, but it cannot start new ones: all events were already running at the start, thus $\sint{\pom} = \pom$.
Note that the discrete pomset $\pomset{\pinterface a \ibullet}{\ibullet b \pinterface}$ is neither a starter nor a terminator, since it both starts~$a$ and terminates~$b$.
A conclist has no interface, it simply denotes an (ordered) list of events that are running concurrently, hence the name, short for concurrency list.

We denote the set of conclists by $\conc$.
We write $\stSet$ for the set of all starters and terminators, and $\stset{k}$ for the ones of dimension at most $k$.
Notice that since the alphabet $\alphabet$ is finite, the set $\stset{k}$ is also finite.
A finite word $\pom_1 \pom_2  \cdots \pom_n \in \stset{k}^*$ is called \defs{coherent} if $\tint{P_i} = \sint{P_{i+1}}$ for all $1 \leq i \leq n-1$.
When that is the case, we can glue the successive elements in the sequence to obtain a pomset $\pom_1 \glue \pom_2 \glue \cdots \glue \pom_n \in \iipomsk{k}$.
A coherent word on $\stset{k}$ is also called an \defs{ST-sequence}.
If $w \in \stset{k}^\ast$ is an ST-sequence, we write $\gluew{w} \in \iipomsk{k}$ for its associated pomset.

\begin{proposition}[ST decomposition \cite{FahrenbergZ23}] \label{def-ST-sequence}
Every pomset $\pom \in \iipomsk{k}$ can be decomposed as an ST-sequence: there exists $w \in \stset{k}^* $ such that $\pom = \gluew{w}$.
\end{proposition}
Let us take our running example of \Cref{fig-pomset-4} and express an ST-decomposition of this pomset:
\[
\begin{tikzpicture}[scale=0.9, baseline=-0.5ex, every node/.style={transform shape}]
	\coordinate (A) at (0,0.5);
	\node (a) at (A) {$\vphantom{\ibullet}a$};
	\node[right of = a, xshift = 0.2cm] (c) {$c\vphantom{\ibullet}$};
	\node[below of = a, yshift = 0.05cm] (b) {$\ibullet b$};
	\node[right of = b, xshift = 0.2cm] (d) {$d\ibullet$};
	\path[precedence] (b) edge (d) (a) edge (c) (b) edge (c);
	\path[event-order] (a) edge (b) (a) edge (d) (c) edge (d);
\end{tikzpicture}
\quad=\quad
\pomset{  \pinterface a \ibullet }{ \ibullet b \ibullet } 
\glue
\pomset{ \ibullet a \ibullet  }{\ibullet b \pinterface }
\glue
\pomset{ \ibullet a \ibullet  }{\pinterface d \ibullet }
\glue
\pomset{ \ibullet a \pinterface  }{\ibullet d \ibullet }
\glue
\pomset{ \pinterface c \ibullet }{\ibullet d \ibullet }
\glue
\pomset{ \ibullet c \pinterface}{\ibullet d \ibullet}
\]
An ST decomposition is called \defs{sparse} if it alternates between starters and terminators, and contains no identities.
For example, the ST decomposition given above is sparse: from left to right, we first start $a$, terminate $b$, start $d$, terminate $a$, start $c$, and terminate $c$.

In general, ST decompositions are not unique. Indeed, one can always add any number of identities (see example on the left); and when several events start at the same time, we can equivalently start one before the other, or both at once (see example on the right).
\[
\hfill
\pomset{a \ibullet}{b \ibullet}
\;=\;
\pomset{a \ibullet}{b \ibullet}
\glue
\pomset{\ibullet a \ibullet}{\ibullet b \ibullet}
\glue
\pomset{\ibullet a \ibullet}{\ibullet b \ibullet}
\hfill
\pomset{a \ibullet}{b \ibullet}
\;=\;
\pomset{a \ibullet}{}
\glue
\pomset{  \ibullet a \ibullet }{ \pinterface b \ibullet }
\;=\;
\pomset{}{b \ibullet}
\glue
\pomset{ \pinterface a \ibullet }{ \ibullet b \ibullet }
\hfill
\]
However, every pomset admits a unique sparse ST decomposition (see~\cite{FahrenbergZ23} for a proof).

%% file: sub/sub-defs/MSO.tex
We now recall the Monadic Second Order 
(\mso) logic over pomsets introduced in~\cite{AmraneBFF24}.
The main result of~\cite{AmraneBFF24} is a variant of Büchi's theorem for pomsets, which states that \mso logic captures the same class of pomset languages as higher dimensional automata.

    The syntax of \mso formulas for pomsets is generated by the grammar:
    \[
    \form,\psi ::= \neg \form \sep \form \wedge \psi \sep \exists x.\form \sep \exists X.\form \sep x \in X \sep
            a(x) \sep \starter(x) \sep \terminator(x) \sep x\precedence y \sep x\eventorder y
    \]
    where $a \in \Sigma$ is a letter of the alphabet, $x,y$ are first order variables and $X$ is a second-order variable.
    The symbols $\starter$ and $\terminator$ are unary predicates, meaning that $x$ belongs to the starting (resp.\ terminating) interface.
    The binary relation symbols $\precedence$ and $\eventorder$ stand for the precedence and event order.

\begin{definition}[Semantics of \mso over pomsets]
	\label{def:MSO-semantics}
    An \mso formula $\form$ is evaluated over a pomset $\pom = (\pom, \precedence_{\pom}, \eventorder_{\pom}, \sint{\pom}, \tint{\pom}, \labelling_{\pom})$, together with an interpretation function $\nu$.
    The function $\nu$ gives the interpretation of free variables of $\form$: first-order variables are mapped to events of~$P$, and second-order variables are mapped to sets of events of~$P$.
	The satisfaction relation $P,\nu \models \form$ is defined inductively as follows:
    \begin{align*}
        &P, \nu \models a(x)  \iif \labelling_P (\nu(x)) = a &
        &P, \nu \models x \in X  \iif \nu(x) \in \nu(X)\\
        &P, \nu \models \starter(x) \iif \nu(x) \in \sint{\pom} &
        &P, \nu \models \terminator(x)  \iif \nu(x) \in \tint{\pom} \\
        &P, \nu \models \neg \form \iif P, \nu \not\models \form &
        &P, \nu \models \form \wedge \psi \iif  P, \nu \models \form \aand P, \nu \models \psi \\
        &P, \nu \models x \precedence y  \iif \nu(x) \precedence_{\pom} \nu(y) &
        &P, \nu \models x \eventorder y  \iif \nu(x) \eventorder_{\pom} \nu(y)\\
        &P, \nu \models \exists x.\form \iif \exists p \in P\; \suchthat P, \nu[x \mapsto p] \models \form \\
        &P, \nu \models \exists X.\form \iif \exists Q \subseteq P\; \suchthat P, \nu[X \mapsto Q] \models \form
    \end{align*}
    We write $P \models \form$ when $\form$ does not have any free variables, and $\langof{\form} = \{P \in \iipoms \mid P \models \form \}$.
\end{definition}

In order to prove that \mso over pomsets is as expressive as higher dimensional automata, the authors of \cite{AmraneBFF24} used a detour via ST-sequences.
An \mso formula $\form$ over pomsets can be translated into a formula $\translated{\form}$ over ST-sequences that accepts the representations of the pomsets accepted by $\form$.
The precise statement of this translation is reproduced below, in \Cref{lemma:mso_translation}.
Recall that ST-sequences are simply words over a different alphabet $\stset{k}$, so \mso logic over ST-sequences is the standard \mso logic over finite words. %

\begin{lemma}[cf.~\protect{\cite[Lemma 12]{AmraneBFF24}}]\label{lemma:mso_translation}
    Let $\form$ be an \mso formula for pomsets without free variables. Then, for any $k \in \mathbb{N}$,
    there exists an \mso formula $\translated{\form}$ over $\stset{k}$, such that:
    \[\langof{\translated{\form}} = \{w \in (\stset{k} \setminus \{\mathrm{Id}_{\emptyset}\})^+ \mid w \textup{ is coherent and } \gluew{w} \models \phi \}\]
    where $\mathrm{Id}_{\emptyset}$ denotes the empty pomset.
\end{lemma}

%% file: sub/fo-pomset-to-fo-st-seq.tex
In this section, our goal is to prove a variant of \Cref{lemma:mso_translation} for \fo formulas.
Unfortunately, the proof given in~\cite{AmraneBFF24} cannot be easily adapted: even if we start with an \fo formula~$\form$, the translation that they give yields a formula $\translated{\form}$ that contains second-order quantification.
This is because their translation makes extensive use of an MSO-definable relation $(x,i) \relation (y,j)$ that keeps track of the position of an event in an ST-sequence.
This relation relies on second-order quantification, and is used in several cases of the inductive translation.
In \Cref{sec-sim-relation}, we show that the relation $\relation$ can actually be expressed in first order.
The translation from \fo formulas on pomsets to \fo formulas on ST-sequences then follows in \Cref{thm-FOP-to-FOST}.

\subsection{First Order logics for pomsets}

First Order (FO) logic is obtained by removing the second-order quantification from the MSO logic described in \Cref{sec:MSO-pomsets}.
Given a finite alphabet $\alphabet$, the syntax of \defs{\fo formulas over pomsets} is generated by the following grammar:
\[
\form,\psi ::= \neg \form \sep \form \wedge \psi \sep \exists x.\form \sep
a(x) \sep \starter(x) \sep \terminator(x) \sep x\precedence y \sep x\eventorder y
\]
where $ x $ is a first order variable and $a \in \alphabet$.
The semantics is the same as the one of \Cref{def:MSO-semantics}, ignoring the two cases related to second-order variables.
When dealing with pomsets of dimension $\leq k$, we also consider \defs{\fo formulas over ST-sequences}.
Recall that an ST-sequence is simply a word over the finite alphabet $\stset{k}$, whose elements are starters/terminators. So this is the usual \fo logic over words, whose syntax is generated by:
\[
\form, \psi ::= \neg\form \mid \form \wedge \psi \mid \exists x.\form \mid P(x) \mid x < y
\]
where $ x $ is a first order variable and $P \in \stset{k}$.
It is interpreted over words $w \in \stset{k}^\ast$, with the usual semantics.
We write $\fok{k}$ for the set of \fo formulas over $\stset{k}$.

\subsection{The same-event relation $\relation$}
\label{sec-sim-relation}

When translating \fo formulas from pomsets to ST-sequences, a key difficulty is that first-order variables contain very different information in those two representations.
In an \fo formula evaluated over a pomset~$P$, a variable $x$ is interpreted as an element $p \in P$ of the pomset, i.e., an event.
However, for an $\fok{k}$ formula evaluated over ST-sequences, a variable $x$ is interpreted as a position in the sequence, labeled by a letter $U \in \stset{k}$.
Each such letter contains several events (up to $k$, the dimension of the language), but it also contains only a small portion of those events.
As seen in the example below, the event $a$ of the pomset depicted on the left spans over 4 letters of the ST-sequence on the right.
\begin{center}
	\scalebox{0.9}{
        \begin{tikzpicture}[baseline=(current bounding box.north), every node/.style={transform shape}]
            \coordinate (A) at (0,0);
            \node (a) at (A) {\color{red} $a$};
            \node[right of = a, xshift = 0.2cm] (c) {$c\vphantom{\ibullet}$};
            \node[below of = a, yshift = 0.05cm] (b) {$\ibullet b$};
            \node[right of = b, xshift = 0.2cm] (d) {$d\ibullet$};
            \path[precedence] (b) edge (d) (a) edge (c) (b) edge (c);
            \path[event-order] (a) edge (b) (a) edge (d) (c) edge (d);
            \node[above = .2cm of a, xshift=0.55cm] (Pomset) {Pomset};
        \end{tikzpicture}
	}
        \hspace{1cm}
        \unskip\vrule
        \hspace{1cm}
	\scalebox{0.9}{
        \begin{tikzpicture}[baseline=(current bounding box.north), every node/.style={transform shape}]
        \node (ST-seq) at (0,0.5) {ST-Sequence};
        \node[below = .2cm of ST-seq] (b) {$\displaystyle\pomset{\color{red}\pinterface a\ibullet}{\ibullet b\ibullet}
            {\pomset{\color{red}\ibullet a \ibullet}{\ibullet b\pinterface}}
        \pomset{\color{red}\ibullet a \ibullet}{\pinterface d \ibullet}
        \pomset{\color{red}\ibullet a \pinterface}{\ibullet d \ibullet}
        \pomset{\pinterface c \ibullet}{\ibullet d \ibullet}
        \pomset{\ibullet c \pinterface}{\ibullet d \ibullet}$};
	    \end{tikzpicture}
	}
\end{center}

Thus, to keep track of events in an ST-sequence, we will use pairs $(x,i)$ where $x$ is a first-order variable, selecting one starter/terminator in the ST-sequence; and $i$ is the position of the event that we are currently tracking in this starter/terminator.
Note that, since we are interested in pomsets of dimension $\leq k$, there are only finitely many possible values for $1 \leq i \leq k$.
Since the same event may span several starters/terminators, it may be the case that two different pairs $(x,i)$ and $(y,j)$ designate the same event, as in the example below:
\begin{align*}
	\pomset{\pinterface{\colora}\ibullet}{\ibullet {\colorb}\ibullet}
    \overbrace{\pomset{\ibullet {\colora} \ibullet}{\ibullet {\colorb}\pinterface}}^{(x,1)}
    \pomset{\ibullet {\colora} \ibullet}{\pinterface {\colord} \ibullet}
    \overbrace{\pomset{\ibullet {\colora} \pinterface}{\ibullet {\colord} \ibullet}}^{(y,1)}
    \pomset{\pinterface {\colorc} \ibullet}{\ibullet {\colord} \ibullet}
     \pomset{\ibullet {\colorc} \pinterface}{\ibullet {\colord} \ibullet}
\end{align*}

Hence, we will need to use the \defs{same-event relation} $(x,i) \relation (y,j)$, which is true
if and only if the $i$-th event of the evaluation of $x$ is the $j$-th event of the evaluation of $y$.
This is the same idea as in the proof of \Cref{lemma:mso_translation} found in~\cite{AmraneBFF24}; however, we will need to show that this relation $\relation$ can actually be defined using only first order formulas.
We do so in \Cref{lem-sim-FO},
by providing a counter-free automaton recognizing it.
Counter-free automata are restrictions of finite state automata whose languages are first-order definable.

\begin{definition}[\cite{McNaughton71}, Counter-Free Automata]
\label{def:counter-free}
Let $ \calA = (\alphabet, Q, q_0, F, \delta) $ be a deterministic finite-state automaton.
$ \calA $ is \defs{counter-free} if there exists a positive integer $ n $ such that 
for any non-empty word $ w \in \alphabet^* $ and for any state $q \in Q$,
the state-transition function %
satisfies the following equality: $\delta(q, w^n) = \delta(q, w^{n+1})$.
\end{definition}

\begin{proposition}[\cite{McNaughton71}]\label{prop-aperiodic-autom}
    Counter-free automata are as expressive as \fo logic over words.
\end{proposition}

\begin{lemma}\label{lem-sim-FO}
    \input{sub/prop/sim-FO.tex}
\end{lemma}

\begin{proof}[Proof (Sketch)]
Given an event $a$ and two indexes $i,j \in \{1, \dots, k\}$, we build a counter-free automaton $\automata_{i,j,a}$ that
scans an
ST-sequence $P_1 P_2 \cdots P_n$ and accepts if and only if the $i$-th element of $P_1$ is the same $a$-event as the $j$-th element of $P_n$.
\Cref{fig-automata-sim} depicts an example of such an automaton, with event alphabet $\Sigma = \{a, b\}$ and pomsets of dimension at most $k = 2$.
The $a$-event that the automaton is currently following is depicted in red on each starter/terminator that it reads.
Inside the states, we keep track of the list of currently active events (\ie a
conclist $U \in \concb{k}$), and the position ($i \in \{1, \dots, k\}$) of the
followed event.
Hence, the set of states is $\concb{k} \times \{1, \dots, k\}$, plus an initial, a final and a sink state. See \Cref{app:proof-lem-sim-FO} for the precise formal definition of the automaton.

\begin{figure}
    \centering
    \input{figures/aut_sim.tex}
    \caption{$\automata_{1,2,a}$, sink state and identities not drawn.}
    \label{fig-automata-sim}
\end{figure}

To prove that all $\automata_{i,j,a}$ are counter-free, consider a state
$(U, \ell)$, and an ST-sequence $w = P_1 \cdots P_n$.
We need to show that $\automata_{i,j,a}$ will never fall in a non-trivial cycle
when reading $w$ repeatedly from $(U,\ell)$. There are several cases:
\begin{itemize}
    \item If $\sint{P_1} \neq U$ or $\tint{P_n} \neq U$, then the execution fails
    and falls in a sink state after one or two iterations.
    \item If $\sint{P_1} = \tint{P_n} = U$ and the $\ell$-th element of $P_1$
    is the $\ell$-th element of $P_n$, then we are in a trivial cycle (an execution reading
    $w$ from $(U,\ell)$ arrives in $(U,\ell)$)
    \item If $\sint{P_1} = \tint{P_n} = U$ and the $\ell$-th element of $P_1$
    is the $\ell'$-th element of $P_n$, with $\ell > \ell'$ (the opposite case
    is similar), then $\ell$ will keep decreasing strictly as we iterate~$w$.
    The example below illustrates what might happen, with a pomset of dimension $k=3$. We start in state $(aaa,3)$, so three $a$'s are running concurrently, and we are tracking the third one (in red).
    When the automaton reads the word $w$ below, the first $a$ is terminated, and another $a$ is started, but this new $a$ is placed after the other two according to event order.
    The tracked $a$ ends up in position 2, so the state of the automaton is now $(aaa,2)$.
    After decreasing a finite number of times, the execution arrives
    in the sink state -- and fails.
    {\[
        w = \left[\begin{array}{ccc}
        \ibullet a \pinterface \\ \color{blue} \ibullet a \ibullet \\ \color{red} \ibullet a \ibullet
        \end{array}\right]
        \left[\begin{array}{ccc}
        \color{blue} \ibullet a \ibullet \\ \color{red} \ibullet a \ibullet \\ \pinterface a \ibullet
        \end{array}\right]
    \qquad\qquad
        \begin{array}{l}
        \delta((aaa, 3), w) = (aaa, 2)\\[6pt]
        \delta((aaa, 3), w^2) = (aaa, 1)
        \end{array}
        \]}

\end{itemize}
Hence, from any state $(U,\ell)$, the execution enters a trivial cycle
in at most $k+1$ steps. Since any execution leaves the initial state in
one step, $\automata_{i,j,a}$ is counter-free (by taking $n = k+2$ in \cref{def:counter-free}).
So $\automata_{i,j,a}$ can be expressed as an \fo formula, and with a disjunction over all $a \in \Sigma$ we get an \fo formula for $(x,i) \sim (y,j)$. See \Cref{app:proof-lem-sim-FO} for details.
\end{proof}

\subsection{Translation of \fo formulas from pomsets to ST-sequences}

Now that we have proven that the same-event relation $\relation$ is \fo-definable, we can inductively translate \fo formulas on pomsets to $\fok{k}$ formulas on ST-sequences. %
The inductive definition of $\translated{\form}$ is the same as the one of~\cite{AmraneBFF24}.
We reproduce it in \Cref{app:proof-thm-FOP-to-FOST} for completeness.

\begin{theorem}\label{thm-FOP-to-FOST}
    \input{sub/prop/FOP-to-FOST.tex}
\end{theorem}

\noindent
In \Cref{app:proof-thm-FOP-to-FOST}, we show that in the worst case,
the size of $\translated{\form}$ is $O(k^{|\form|})$, where $|\form|$ is
the size of $\form$.
As a corollary of \Cref{thm-FOP-to-FOST}, we can can show that \fo is strictly weaker than \mso for pomset languages.

\begin{corollary}\label{cor-FO<MSO}
    \fo is strictly weaker than \mso.
\end{corollary}

\begin{proof}
    Consider the language of pomsets $L = \left\{\pomset{a}{b}^{2n}  \sep n \in \mathbb{N}\right\}$, where the exponent denotes gluing iteration.
    Viewed as a language of ST-sequences (i.e., words on $\stset{2}$), this corresponds to
    $W = \left\{\left(\pomset{a\interface}{b\interface} \pomset{\interface a}{\interface b}\right)^{2n}\mid n \in \mathbb{N}\right\}$.
    As a language of words, $W$ is not definable in $\fok{k}$
    (see e.g.~\cite{DiekertG08} for a proof).
    By the contrapositive of \Cref{thm-FOP-to-FOST}, $L$ cannot be defined in~\fo.
\end{proof}

%% file: sub/prop/sim-FO.tex
Fix $i,j \in \{1, \dots, k\}$. Then, the binary relation $(x, i) \relation (y,j)$
is definable by an $\fok{k}$ formula with two free variables $x$ and $y$.

%% file: figures/aut_sim.tex
\begin{tikzpicture}
    [auto, node distance=2cm, initial text=, every state/.style={minimum size = 1.1cm}]
        \node[state, initial]   (bot)                       {$\bot$};
        \node                   (1)   [right = of bot]       {};
        \node[state]            (a1)  [right = 1.7cm of 1]      {$a,1$};
        \node[state]            (aa1) [below = 1cm of 1]       {$aa,1$};
        \node[state]            (ab1) [above = 1cm of 1]       {$ab,1$};
        \node                   (2)   [right = of a1]       {};
        \node[state, accepting] (aa2) [below  = 1cm of 2]    {$aa,2$};
        \node[state, accepting] (ba2) [above = 1cm of 2]    {$ba,2$};
        \node[state, accepting] (top) [right = 1.7cm of 2]  {$\top$};

        \path[-latex, inner sep=1pt]
            (bot)   edge                    node [swap]         {$\pomset{(\interface)\follow{a}\interface}{(\interface)a\interface}$\hphantom{-}} (aa1)
                    edge                    node [pos=0.3]         {$[\follow{a} \interface]$} (a1)
                    edge                    node []        {\hphantom{-}$\pomset{(\interface)\follow{a}\interface}{(\interface)b\interface}$} (ab1)
            (aa1)   edge [bend left=.4cm]   node []        {$\pomset{\interface \follow{a} \interface}{\interface a \pinterface}$} (a1)
            (ab1)   edge [bend left=.4cm]   node []        {$\pomset{\interface \follow{a} \interface}{\interface b \pinterface}$} (a1)
            (a1)    edge [bend left=.4cm]   node []        {$\pomset{\interface \follow{a} \interface}{\pinterface a \interface}$} (aa1)
                    edge [bend left=.4cm]   node []        {$\pomset{\interface \follow{a} \interface}{\pinterface b \interface}$} (ab1)
                    edge [bend right=.4cm]  node [swap]         {$\pomset{\pinterface a\interface}{\interface \follow{a} \interface}$} (aa2)
                    edge [bend right=.4cm]  node [swap]   {$\pomset{\pinterface b\interface}{\interface \follow{a} \interface}$\hphantom{\quad}} (ba2)
            (aa2)   edge [bend right=.4cm]  node [swap]        {$\pomset{\interface a \pinterface}{\interface \follow{a} \interface}$} (a1)
                    edge                    node [swap]         {$\pomset{\interface a (\interface)}{\interface \follow{a} \phantom{(\interface)}}$} (top)
            (ba2)   edge [bend right=.4cm]  node [swap]        {$\pomset{\interface b \pinterface}{\interface \follow{a} \interface}$} (a1)
                    edge                    node []        {$\pomset{\interface b (\interface)}{\interface \follow{a} \phantom{(\interface)}}$} (top);
\end{tikzpicture}

%% file: sub/prop/FOP-to-FOST.tex
Let $\form$ be an \fo formula over pomsets without free variables.
Then, for any $k \in \mathbb{N}$, there exists an $\fok{k}$ formula $\translated{\form}$ over $\stset{k}$ such that:
\[
\langof{\translated{\form}} = \{w \in (\stset{k} \setminus \{\mathrm{Id}_{\emptyset}\})^+ \mid w \textup{ is coherent and } \gluew{w} \models \form \}
\]

%% file: sub/ltl-pomsets.tex
In this section, we introduce our Linear Temporal Logic for pomsets, \lsptl.
Using the terminology of~\cite{DiekertG06local-global}, it is a \emph{local} temporal logic, meaning that it is interpreted over a single event of the pomset, rather than on a global state.
However, since events can span a long period of time, they are split into so-called \emph{sub-events} (see \Cref{fig-ltl-example}).
While defining a temporal logic directly over events might seem more straightforward, we show in \Cref{app:sub-events} that such a
logic is strictly less expressive than \fo over pomsets.
Nonetheless, we will show in \Cref{sec-ltl-expressivity}, \Cref{thm-LTLP--FOP} that with sub-events, \lsptl is as expressive as \fo over pomsets of dimension $\leq k$.

\subsection{Syntax and Semantics}
\label{sec-ltl-def}

Let us first give an intuitive explanation of the notion of sub-event, relying on the interval representation of a pomset. 
Consider the pomset $P$ depicted in \Cref{fig-ltl-example}.
The interval representation of $P$ is decomposed in three ``slices'' (delimited by dotted lines).
Each slice corresponds to two consecutive elements of the sparse ST-decomposition of $P$, i.e., a starter and a terminator.
Thus, in every slice, some events start, then some events terminate.
Intuitively, a sub-event is given by an event $x \in P$ together with one of the slices it crosses.
For instance, event $a$ spans two slices: therefore, it is divided into two sub-events $a_1$ and $a_2$.

\begin{figure}[h]
    \centering
    \input{figures/ltl-example.tex}
    \caption{A slicing of a pomset with 4 events ($a$, $b$, $c$, $d$) into 6 sub-events ($a_1$, $a_2$, $b_1$, $c_1$, $d_1$, $d_2$)}
    \label{fig-ltl-example}
\end{figure}

Our formal definition of sub-events does not rely on the ST-decomposition of the pomset.
Instead, we define it directly on the pomset itself.
A sub-event is a pair of two events, $(x,m)$, where $x \in P$ is the event being considered, and $m \in P$ acts as a timestamp indicating the current slice.
One can think of $m$ as \emph{the latest event to have started}.
For example, in \Cref{fig-ltl-example}, the sub-event $a_1$ is given by the pair $(a,a)$; while $a_2$ is given by the pair $(a,d)$.
However, notice that not all pairs of events define sub-events: $(c,d)$ makes no sense since event $c$ is not running when $d$ starts.
\Cref{def-order-start} introduces an order that captures these situations.

\begin{definition}
    \label{def-order-start}
	Let $P$ be a pomset.
    We say that $x \in P$ \defs{starts before} $y\in P$, denoted $x <^s y$ if there exists $z \in P$ such that $x \parallel z$ and $z < y$.
    We write $x \sim^s y$ when $x \not<^s y$ and $y \not<^s x$, and $x \lesssim^s y$ when $x <^s y$ or $x \sim^s y$.
\end{definition}

Intuitively, $x <^s y$ means that in \emph{every}
interval representation of $P$, the interval representing~$x$ starts before the one representing~$y$.
For example in the pomset of \Cref{fig-ltl-example}, we have $a <^s d$ and $d <^s c$. Observe that $\sim^s$ is an equivalence relation whose equivalence classes correspond intuitively to the slices of~$P$ (for example, $a \sim^s b$ in \Cref{fig-ltl-example}).
Hence, $\lesssim^s$ is a total preorder on the events of $P$, and we think of it as a total order on the slices.

\begin{definition}[Sub-event]
    \label{def-subevent}
    A \defs{sub-event} of a pomset $P$ is a pair $(x,m) \in P^2$ where $x \parallel m$ and $x \lesssim^s m$.
    When $x = y$ and $m \sim^s q$, we write $(x,m) \equiv (y,q)$.
\end{definition}

Note that, when several events start in the same slice, two sub-events may actually represent the same point in the interval representation.
This is why we introduce the $\equiv$ relation over sub-events.
When $(x,m) \equiv (y,q)$, we think of $(x,m)$ and $(y,q)$ as representing ``the same'' sub-event.
For example, the sub-event denoted by $a_1$ in \Cref{fig-ltl-example} can be formalized either by $(a,a)$ or $(a,b)$, since $a \sim^s b$.
To simplify the presentation, we do not explicitly quotient by the relation~$\equiv$.
However, we will make sure that whenever $(x,m) \equiv (y,q)$, the two sub-events satisfy the same formulas, i.e., $P, (x,m) \models \phi$ iff $P,(y,q) \models \phi$ (see \Cref{prop-ltlp-sube}).

Next, we define an order between sub-events.
This order will be crucial to ensure that our temporal logic operators can only go forward in time.
There are two ways to advance in time: either stay within the same event~$x$, but move to a later slice; or terminate the current event~$x$ and jump to a new event~$y$ that occurs after~$x$.

\begin{definition}[Order over sub-events]
    \label{def-se-order}
    Given two sub-events $(x,m)$ and $(y,q)$, we say that $(x,m)$ \defs{precedes} $(y,q)$ in $P$, denoted $(x,m) \prec (y,q)$, if either $x < y$, or $x = y$ and $m <^s q$.
\end{definition}

We denote by $\preccurlyeq$ the non-strict version: $(x,m) \preccurlyeq (y,q)$ if $(x,m) \prec (y,q)$ or $(x,m) \equiv (y,q)$.
Observe that $\preccurlyeq$ is a (partial) preorder on sub-events.
In \Cref{fig-ltl-example}, we have for example $a_1 \prec a_2 \prec c_1$ and $b_1 \prec c_1$.
However, $a_1 \not\prec d_1$ and $b_1 \not\prec a_2$: it is not allowed to jump to a different event that is concurrent with the current one.
Indeed, we want our temporal operators to be able to follow the local view of individual processes. (There will, however, be a different modality allowing to jump to a concurrent event within the same slice.)

Finally, we introduce a one-step version of $\prec$, which plays the role of our ``next'' modality.

\begin{definition}
    \label{def-se-order1}
    $(x,m) \prec_1 (y,q)$ if $(x,m) \prec (y,q)$ and there is no $r$ such that $m <^s r <^s q$.
\end{definition}

Intuitively, the relation $\prec_1$ only orders events in adjacent slices. For instance, in \Cref{fig-ltl-example}, we have $b_1 \prec_1 d_1$, but $b_1 \not \prec_1 c_1$.
Note that \Cref{def-se-order1} is more restrictive than the requirement that ``there is no $(z,r)$ such that $(x,m) \prec (z,r) \prec (y,q)$''.
The latter would allow jumping directly from $b_1$ to $c_1$ in \Cref{fig-ltl-example}, skipping a slice.
We prefer to ensure that our ``next'' modality always advances by exactly one slice.

We are now ready to introduce our logic \lsptl.
The ``next'' modality, denoted $\enext$, jumps to a successor sub-event in the next slice.
Note that it is existential, since there may be more than one successor.
The modalities $\ltlcirc$ and $\ltlrcirc$ jump to a concurrent sub-event within the same slice. Recall that in our pomsets, concurrent events are totally ordered by the event order $\eventorder$.
Finally, the atomic formula $\starter$ (resp.\ $\terminator$) checks whether the current sub-event is being started (resp.\ terminated) in the current slice.

\begin{definition}[\lsptl]
    \label{def-ltlp}
    The syntax of \lsptl is generated by the following grammar:
    \[\form, \psi ::= a \mid \starter \mid \terminator \mid \neg\form \mid \form \wedge \psi \mid \enext \form \mid \ltlcirc \form \mid \ltlrcirc \form \mid \form \until \psi\]
    where $a \in \alphabet$.
    Given a \lsptl formula $\phi$ and a sub-event $(x,m)$, we define the satisfaction relation $P,(x,m) \models \phi$ by induction on the formula $\phi$:
    \begin{enumerate}
        \item $P, (x,m) \models a$ if $\labelling_P(x) = a$
        \item $P, (x,m) \models \starter$ if $x \not\in \sint{P}$ and $x \sim^s m$ %
        \item $P, (x,m) \models \terminator$ if $x \not\in \tint{P}$ and there is no $q \in P$ s.t.\ $x \parallel q$ and $m <^s q$
        \item $P, (x,m) \models \neg\form$ if $P, (x,m) \not\models \form$
        \item $P, (x,m) \models \form \wedge \psi$ if $P, (x,m) \models \form$ and $P, (x,m) \models \psi$
        \item $P, (x,m) \models \enext \form$ if there is a sub-event $(y,q)$ such that $(x,m) \prec_1 (y,q)$ and ${P, (y,q) \models \form}$.
        \item $P, (x,m) \models \ltlcirc \form$ if there is $y \in P$ such that $(y,m)$ is a sub-event, $y$ is a direct successor of $x$ by $\eventorder$, and $P, (y,m) \models \form$
        \item $P, (x,m) \models \ltlrcirc \form$ if there is $y \in P$ such that $(y,m)$ is a sub-event, $y$ is a direct predecessor of $x$ by $\eventorder$, and $P, (y,m) \models \form$
        \item $P, (x,m) \models \form \until \psi$ if there exists $(y, n)$ such that $(x,m) \preccurlyeq (y, n)$, $P, (y,n) \models \psi$ and, for all $(z,q)$ such that
        $(x,m) \preccurlyeq (z,q) \prec (y,n)$, $P, (z,q)\models \form$.
    \end{enumerate}
\end{definition}

To define what ``$P$ satisfies $\form$''
means for a given pomset~$P$ and  \lsptl formula~$\form$, we need to choose a canonical ``source'' sub-event of the pomset. 
Let $M_P \subseteq P$ be the set of events that are minimal according to $<_P$.
Notice that $M_P$ is totally ordered by the event-order relation $\eventorder_{P}$.
Then we let $\src{P} = (x,x)$, where $x = \min_{\,\eventorder_{P}} M_P$.
Intuitively, this is the top-left sub-event in an interval representation. %
Finally, define $P \models \phi$ iff $P, \src{P} \models \phi$.

\begin{example}
\begin{figure}[h]
    \centering
	\vspace{-0.2cm}
    \input{figures/ltl-next.tex}
    \caption{A pomset containing 5 events $a$, $b$, $c$, $d$, $e$, where $c$ is split into two sub-events $c_1$ and~$c_2$. Note that formally, $c_1$ has two representatives $(c,a) \equiv (c,c)$, while $c_2$ is represented by $(c,b)$.}
    \label{fig-next}
\end{figure}

In the pomset represented in \Cref{fig-next}, $c_1 \models \starter \wedge \neg \terminator$ while $c_2 \models \neg \starter \wedge \terminator$.
Notice that $a_1 \not\models \starter$: since event~$a$ belongs to the starting interface of the pomset, it was already started in the first slice. Similarly, $e_1 \not\models \terminator$.
The $\enext$ modality allows to jump from $c_2$ to either $d_1$ or $e_1$, with an existential quantification.
It can be read as an ``exists next'' modality.
For instance, we have $c_2 \models \enext d$ and $c_2 \models \enext e$, but $c_2 \not\models \enext d \wedge e$.
However, note that $c_1 \not\models b$: the ``next'' modality does not allow to jump to a different event since event $c$ is still running (cf.\ \Cref{def-se-order1}).
The dual ``for all next'' modality can be defined as $\anext \form := \neg \enext \neg\form$.

The operators $\ltlcirc$ and $\ltlrcirc$ allow to move to a concurrent event following the event order relation, while staying within the current slice.
So, crucially, $b_1 \not\models \ltlcirc \ltlrcirc a$.
This was our motivation to introduce the notion of sub-events. Without them, one could inadvertently move forwards or backwards in time by only following event order relations.
Note that, within a slice, the event order relation is total. Thus, there can be at most one direct successor. The operator $\ltlcirc$ is still existential in a degenerate sense: if there is no successor, the formula is not satisfied. For instance, $e_1 \not\models \ltlcirc \top$.

The operator $\until$ is the usual Until modality with regard to $\preccurlyeq$.
However, it might be slightly counter-intuitive. Despite the universal quantification over all intermediary sub-events, it may seemingly ``miss'' events that are concurrent with one of the two endpoints.
For instance, in the pomset of \Cref{fig-next}, $a_1 \models (a \vee b) \until d$: since event~$c$ is parallel with~$a$, the sub-event $c_2$ is not reachable from $a_1$.
In the next section, we will define a variant of the Until operator that also takes into account these parallel events.
\end{example}

As stated earlier, our logic is consistent with the $\equiv$ relation. This is proven by induction
over the formula; the full proof can be found in \Cref{app:proof-prop-ltlp-sube}.

\begin{proposition}
    \label{prop-ltlp-sube}
    \input{sub/prop/ltlp-sube.tex}
\end{proposition}

\subsection{Derived operators}

The precise choice of operators for the logic \lsptl may seem somewhat arbitrary. It will be justified in \Cref{sec-ltl-expressivity}, where we show that it is equivalent to \fo.
So, any first-order definable temporal operators can also be defined in \lsptl.
We give a few useful examples below.
Recall that in the whole paper, we are working with pomsets of bounded dimension~$k$. So there can be at most $k$ parallel events at any time.

\begin{itemize}
    \item \textbf{Exists parallel:} $\ltlpar \form := \bigvee_{i = 0}^{k} (\ltlcirc^i \form \vee \ltlrcirc^i \form)$.
    This operator jumps to some sub-event in the current slice.
    $P, (x,m) \models \ltlpar \form$ iff there exists $y \in P$ such that $(y,m)$ is a subevent and $P,(y,m) \models \form$. The dual universally quantified operator is $\ltlapar \form := \neg\ltlpar\neg\form$.
    \item \textbf{Exists strict parallel:} $\ltlpars \form := \bigvee_{i = 1}^{k} (\ltlcirc^i \form \vee \ltlrcirc^i \form)$.
    Similar to the previous one, but this operator must jump to a different sub-event.
    In particular, if there is no other event currently running, this formula is always false.
    \item \textbf{Finally and Globally:} $\finally \form := \top \until \form$. This is the usual Finally operator of temporal logics: $P, (x,m) \models \finally \form$ iff there exists a sub-event $(y,q)$ of $P$ such that $(x,m) \preccurlyeq (y,q)$ and $P, (y,q) \models \form$.
    Its dual operator is Globally, defined as $\globally \form := \neg\finally\neg\form$.
    \item \textbf{Finally parallel and Globally parallel:} $\finallypar \form := \finally \ltlpar \form$.
    This variant of $\finally$ covers all sub-events that happen later or at the same time as the current sub-event, even those that are on parallel events.
    $P, (x,m) \models \finallypar \form$ iff there exists a sub-event $(y,q)$ of~$P$ such that $m \lesssim^s q$ and $P, (y,q) \models \form$.
    The dual universally quantified operator is $\globallypar \form := \neg \finallypar \neg \phi$. Equivalently, one can also define it as  $\globallypar \form := \globally \ltlapar \form$.
    \item \textbf{Until parallel:} $\form \untilp \psi := (\ltlapar \form) \until (\ltlpar \psi)$. This variant of the Until modality takes into account events that are concurrent with the two endpoints. 
	More formally,
    $P, (x,m) \models \form \untilp \psi$ iff there exists a sub-event of the form $(y,q)$ with $m \lesssim^s q$ such that $P,(y,q)\models \psi$ and
    for every sub-event $(z,r)$ such that $m \lesssim^s r <^s q$, $P,(z,r)\models \form$.
\end{itemize}

\subsection{Toy example}

To illustrate how our logic \lsptl allows to concisely express properties of concurrent systems, let us consider a very simple example: specifying the correctness of a mutual exclusion algorithm using locks.
Thus, suppose that we have a lock mechanism available, with two operations: action $P$ to \emph{acquire} the lock, and action $V$ to \emph{release} it~\cite{Dijkstra68}.
Now assume that we want our processes to use some critical resource~$a$ (perhaps a shared data structure) that cannot be accessed concurrently.
So, we want to ensure that there can never be two processes executing action~$a$ concurrently.
The obvious solution to this problem is the following: every process runs the program $(P;a;V)^\ast$. That is, first acquire the lock, then perform action~$a$, and then release the lock (and repeat).

We would like to specify that this implementation of a critical section is correct.
For that, we need to specify (\ref{enum:i}) the behavior of the lock, and (\ref{enum:ii}) the expected behavior of the mutual exclusion algorithm.
We express both of those properties in the language of \lsptl. %

\begin{enumerate}[(i)]
\item \label{enum:i}
\textbf{Lock specification.}
$\phi_{\text{lock}} = \globallypar ((P \wedge \terminator) \Rightarrow (\neg\ltlpars(P \wedge \terminator)  \wedge \enext(\neg(P \wedge \terminator) \untilp (V \wedge \starter))))$.

This formula expresses that at any point during the execution of the program, whenever an action $P$ terminates (i.e., a process acquires the lock), then no other process can acquire the lock ($P \wedge \terminator$) until the lock is released ($V \wedge \starter$).
\item \label{enum:ii}
\textbf{Mutual exclusion specification.}
$\phi_{\text{exclusion}} = \globallypar (a \Rightarrow \neg\ltlpars a)$.

This formula expresses that there are never two overlapping events~$a$.
\end{enumerate}

What we want to verify is that, assuming the $P/V$ actions behave according to the lock specification, our algorithm ensures the mutual exclusion property is satisfied.
The algorithm itself (say, for $k$ processes) can be modeled as an HDA, obtained as the parallel composition of $k$ copies of one process performing a loop $(P;a;V)^\ast$.
Executions of this HDA are pomsets, but not all of them satisfy the lock specification. We want to check that every execution that satisfies $\phi_{\text{lock}}$ also satisfies $\phi_{\text{exclusion}}$.
This amounts to model-checking that every pomset in the language of the HDA satisfies the formula $\phi_{\text{lock}} \Rightarrow \phi_{\text{exclusion}}$.

%% file: figures/ltl-example.tex
\begin{tikzpicture}[baseline=-12ex,scale=1, every node/.style={transform shape}]
    \coordinate (A) at (0,0);
    \node (a) at (A) {$a$};
    \node[right of = a, xshift = 0.2cm] (c) {$c\vphantom{\ibullet}$};
    \node[below of = a, yshift = 0.05cm] (b) {$\ibullet b$};
    \node[right of = b, xshift = 0.2cm] (d) {$d \ibullet$};
    \path[precedence] (b) edge (d) (a) edge (c) (b) edge (c);
    \path[event-order] (a) edge (b) (a) edge (d) (c) edge (d);
\end{tikzpicture}
\hspace{2cm}
\begin{tikzpicture}[scale=1.2]

    \def\possh{-1.3};
    \def\hw{0.3};
    \def\se{1.5};
    \def\margin{0.2};
    
    \filldraw[pomset-1](\margin,0.7)--(2*\se - \margin,0.7)--(2*\se - \margin,0.7+\hw)--(\margin,0.7+\hw)--(\margin,0.7);
    \filldraw[pomset-2](0.0,0.2)--(\se - \margin,0.2)--(\se - \margin,0.2+\hw)--(0.0,0.2+\hw);%
    \filldraw[pomset-4](2*\se + \margin,0.7)--(3*\se - \margin,0.7)--(3*\se - \margin,0.7+\hw)--(2*\se + \margin,0.7+\hw)--(2*\se + \margin,0.7);
    \filldraw[pomset-3](3*\se,0.2+\hw)--(\se + \margin,0.2+\hw)--(\se + \margin,0.2)--(3*\se,0.2);

    \draw[thick,-](0,0)--(0,0.2);
    \draw[thick,-](0,0.2+\hw)--(0,1.2);
    \draw[thick,-](3*\se,0)--(3*\se,0.2);
    \draw[thick,-](3*\se,0.2+\hw)--(3*\se,1.2);

    \draw[dotted,thick,-](\se,0)--(\se,1.2);
    \draw[dotted,thick,-](2*\se,0)--(2*\se,1.2);

    \node at (\se*1/2,0.7+\hw*0.5) {$a_1$};
    \node at (\se*3/2,0.7+\hw*0.5) {$a_2$};
    \node at (\se*1/2,0.2+\hw*0.5) {$b_1$};
    \node at (\se*5/2,0.7+\hw*0.5) {$c_1$};
    \node at (\se*3/2,0.2+\hw*0.5) {$d_1$};
    \node at (\se*5/2,0.2+\hw*0.5) {$d_2$};

    \node at (\se*1/2, -.5) {
        $\pomset{\pinterface{\colora}\interface}{\interface {\colorb}\interface}
        \pomset{\interface {\colora} \interface}{\interface {\colorb}\pinterface}$
    };
    \node at (\se*3/2, -.5) {
        $\pomset{\interface {\colora} \interface}{\pinterface {\colord} \interface}
        \pomset{\interface {\colora} \pinterface}{\interface {\colord} \interface}$
    };
    \node at (\se*5/2, -.5) {
        $\pomset{\pinterface {\colorc} \interface}{\interface {\colord} \interface}
        \pomset{\interface {\colorc} \pinterface}{\interface {\colord} \interface}$
    };

    \node at (\se, -.5) {$*$};
    \node at (2*\se, -.5) {$*$};

\end{tikzpicture}

%% file: figures/ltl-next.tex
\begin{tikzpicture}[scale=1, every node/.style={transform shape}]
    \coordinate (A) at (0,0);
    \node (a) at (A) {$\ibullet a$};
    \node[right of = a, xshift = 0.2cm] (b) {$b$};
    \node[below of = b, yshift = 0.05cm] (c) {$c$};
    \node[right of = b, xshift = 0.2cm] (d) {$d$};
    \node[below of = d, yshift = 0.05cm] (e) {$e \ibullet$};
    \path[precedence] (a) edge (b) (b) edge (d) (b) edge (e) (c) edge (d) (c) edge (e);
    \path[event-order] (a) edge (c) (b) edge (c) (d) edge (e);
\end{tikzpicture}
\hspace{2cm}
\begin{tikzpicture}[scale=1.2]

    \def\possh{-1.3};
    \def\hw{0.3};
    \def\se{1.5};
    \def\margin{0.2};
    
    \filldraw[pomset-4](0,0.7)--(\se - \margin,0.7)--(\se - \margin,0.7+\hw)--(0,0.7+\hw);
    \filldraw[pomset-4](\se + \margin,0.7)--(2*\se - \margin,0.7)--(2*\se - \margin,0.7+\hw)--(\se + \margin,0.7+\hw)--(\se + \margin,0.7);
    \filldraw[pomset-2](\margin,0.2)--(2*\se - \margin,0.2)--(2*\se - \margin,0.2+\hw)--(\margin,0.2+\hw)--(\margin,0.2);
    \filldraw[pomset-1](2*\se + \margin,0.7)--(3*\se - \margin,0.7)--(3*\se - \margin,0.7+\hw)--(2*\se + \margin,0.7+\hw)--(2*\se + \margin,0.7);
    \filldraw[pomset-1](3*\se,0.2+\hw)--(2*\se + \margin,0.2+\hw)--(2*\se + \margin,0.2)--(3*\se,0.2);

    \draw[thick,-](0,0)--(0,0.7);
    \draw[thick,-](0,0.7+\hw)--(0,1.2);
    \draw[thick,-](3*\se,0)--(3*\se,0.2);
    \draw[thick,-](3*\se,0.2+\hw)--(3*\se,1.2);

    \draw[dotted,thick,-](\se,0)--(\se,1.2);
    \draw[dotted,thick,-](2*\se,0)--(2*\se,1.2);

    \node at (\se*1/2,0.7+\hw*0.5) {$a_1$};
    \node at (\se*3/2,0.7+\hw*0.5) {$b_1$};
    \node at (\se*1/2,0.2+\hw*0.5) {$c_1$};
    \node at (\se*5/2,0.7+\hw*0.5) {$d_1$};
    \node at (\se*3/2,0.2+\hw*0.5) {$c_2$};
    \node at (\se*5/2,0.2+\hw*0.5) {$e_1$};

\end{tikzpicture}

%% file: sub/prop/ltlp-sube.tex
Let $P$ be a pomset, and let $(x,m) \equiv (y,q)$ two equivalent sub-events of~$P$, then
for every formula $\form$ of \lsptl, $P,(x,m) \models \form$ if and only if $P,(y,q) \models \form$.

%% file: sub/expressivity.tex
In this section, we show that \lsptl has the same expressive power as first order logic for pomsets of bounded dimension.
We prove this by providing well-chosen translations between pomset and word logics, thus allowing us to use directly the results of Kamp's theorem.
The translation from \lsptl to \fo can be done by writing the semantics of \lsptl in \fo formulas (see \Cref{lem-SPTL-to-FOP}).
Since we already gave the translation from \fo on pomsets to \fo on ST-sequences in \Cref{thm-FOP-to-FOST}, we now have to prove that
any \ltl formula over ST-sequences can be translated to an equivalent \lsptl formula.

\begin{figure}[h]
    \input{figures/translations.tex}
	\vspace{-0.5cm}
    \caption{Summary of the proof of \Cref{thm-sptl=fo}}
    \label{fig-translations-proof}
\end{figure}

First, we formalize the intuition given at the start of \Cref{sec-ltl-def}: the ``slices'' of a pomset can be obtained by gluing two consecutive elements (a starter and a terminator) of its sparse ST-decomposition.
The following definition allows padding the sparse ST-decomposition with identity elements, in order to make sure that it always starts with a starter, and ends with a terminator.
We write $\mathrm{Id}_U$ for the identity pomset with starting and terminating interface~$U$.

\begin{definition}
    Given a pomset $P$ with sparse ST-decomposition $P_1 P_2 \cdots P_n$, the \defs{\evensparse} ST-decomposition of $P$ is $S P_1 P_2 \cdots P_n T$
    where $S = \mathrm{Id}_{S_{P_1}}$ if $P_1$ is a terminator, or $S = \varepsilon$ otherwise, and $T = \mathrm{Id}_{T_{P_n}}$ if $P_n$ is a starter, or $T = \varepsilon$ otherwise.
\end{definition}

\begin{proposition}
    \label{prop-induction}
    \label{prop-LTLST-to-LTLP}
    \label{lem-LTL-to-SPTL}
    \input{sub/prop/ltlp-induction.tex}
\end{proposition}

\begin{proof}[Proof sketch]
	Notice that, since one ``slice'' of a pomset corresponds to two symbols in the ST-decomposition, one application of the Next modality $\enext$ of \lsptl corresponds to two applications of the Next modality $\ltlnext$ of LTL.
	Thus, we actually define two translations $\translated{\form}_\starter$ and $\translated{\form}_\terminator$ by mutual induction.
	The first translation is the one required to prove the Proposition, $\translated{\form} = \translated{\form}_\starter$.
	The second translation $\translated{\form}_\terminator$ is used after one application of $\ltlnext$, when the ST-decomposition starts with a terminator.
	
    For the base case where $\form = P$ for some starter or terminator~$P$, $\translated{\form}_\starter$ checks that~$P$ is indeed a starter,
    and that the elements in the current slice are exactly those of $P$.
	Similarly for $\translated{\form}_\terminator$, we check that $P$ is a terminator.
    To translate the $\ltlnext$ modality, if the current element is a starter, then the next one is the terminator of the same slice: $\translated{\ltlnext\psi}_\starter = \translated{\psi}_\terminator$.
    If, on the contrary, the current element is a terminator, we need to jump to the next slice: $\translated{\ltlnext\psi}_\terminator =
    \enext\translated{\psi}_\starter$.
    For the Until modality, we need to ensure that any intermediate slice satisfies~$\psi_1$ for both the starter and terminator parts,
    until the final slice which satisfies~$\psi_2$ either in its starter or terminator part:
    $\translated{\psi_1 \until \psi_2}_\starter = (\translated{\psi_1}_\starter \wedge \translated{\psi_1}_\terminator)
    \until
    (\translated{\psi_2}_\starter \vee (\translated{\psi_1}_\starter \wedge \translated{\psi_2}_\terminator))$
    and
    $\translated{\psi_1 \until \psi_2}_\terminator =
    \translated{\psi_2}_\terminator \vee 
	(\translated{\psi_1}_\terminator \wedge \enext(
    (\translated{\psi_1}_\starter \wedge \translated{\psi_1}_\terminator)
    \until
    (\translated{\psi_2}_\starter \vee (\translated{\psi_1}_\starter \wedge \translated{\psi_2}_\terminator))))$.

    The full proof is in \Cref{app:proof-prop-induction}. %
\end{proof}

\Cref{lem-LTLP-to-FOP} is proved by expressing the semantics of \Cref{def-ltlp} in \fo. Full details are in \Cref{app:proof-lem-ltlp-to-fop}.
Then we conclude the proof of expressivity of \lsptl in \Cref{thm-LTLP--FOP}.

\begin{proposition}
    \label{lem-SPTL-to-FOP}
    \label{lem-LTLP-to-FOP}
    \input{sub/prop/ltlp-to-fop.tex}
\end{proposition}

\begin{theorem}
    \label{thm-sptl=fo}
    \label{thm-LTLP--FOP}
    For any pomset language $L$ of bounded dimension $k$, the two following statements are equivalent:
    \begin{enumerate}
        \item There exists $\varphi$, an \fo formula, such that for any pomset $P$, $P \in L$ if and only if $P \models \varphi$.
        \item There exists $\psi$, an \lsptl formula, such that for any pomset $P$, $P \in L$ if and only if $P \models \psi$.
    \end{enumerate}
\end{theorem}
\begin{proof}
    Let us proceed as depicted in \Cref{fig-translations-proof}.
    By \Cref{thm-FOP-to-FOST}, any \fo formula $\form$ over pomsets can be translated into an equivalent \fo formula $\form'$ over ST-sequences, accepting exactly the ST-decompositions
    of pomsets accepted by $\form$. This formula can in turn be translated into an equivalent \ltl formula $\form''$ by Kamp's theorem \cite{Rabinovich14}.
    Finally, by \Cref{prop-LTLST-to-LTLP}, there is an \lsptl formula $\form'''$, such that for any pomset $P$,
    $P \models \form'''$ if and only if $P$'s \evensparse ST-decomposition validates $\form''$, which in turn is equivalent to the fact that
    $P \models \form$.
    This proves that \fo is at most as expressive as \lsptl. \Cref{lem-LTLP-to-FOP} ensure that they are in fact equivalent by providing the converse
    translation.
\end{proof}

\begin{remark}
    All our proofs are constructive, in the sense that algorithms can easily be extracted from our translations and from Kamp's theorem \cite{Rabinovich14}.
    Given an \fo formula $\form$, the size of the associated \lsptl formula is non-elementary in the size of $\form$. This is because there is a
    non-elementary succinctness gap between \fo and \ltl over words \cite{Rabinovich14}.
    For the converse translation, given an \lsptl formula $\psi$, the size of the translated formula \fo is linear in the size of $\psi$.
    The complexity for each translation depicted in \Cref{fig-translations-proof} can be found alongside the associated proofs in \Cref{app:proof-thm-FOP-to-FOST,app:proof-prop-induction,app:proof-lem-ltlp-to-fop}.
\end{remark}

%% file: sub/prop/ltlp-induction.tex
For any $k$, for any \ltl formula $\form$ over $k$-dimensional ST-sequences, there exists an \lsptl formula $\translated{\form}$ such that for any pomset $P$
with \evensparse decomposition $P_1 P_2 \cdots P_{2n}$,
we have $P \models \translated{\form}$ if and only if $P_1 P_2 \cdots P_{2n} \models \form$.

%% file: sub/prop/ltlp-to-fop.tex
For any \lsptl formula $\form$, there exists an \fo formula over pomsets $\translated{\form}$ such that, for any pomset $P$,
$P \models \form$ if and only if $P \models \translated{\form}$.

%% file: sub/conclusion.tex
This paper is a first step towards understanding temporal logics on HDA pomset languages.
We have established two key results.
The first one, \Cref{thm-FOP-to-FOST}, shows that \fo formulas on pomsets can be translated to \fo on ST-sequences.
Our second result is \Cref{thm-sptl=fo}, asserting that the temporal logic \lsptl is equivalent to \fo over pomsets, thus extending Kamp's theorem to pomset languages.

As future work, 
it would be insightful to find other characterizations of the class of FO-definable pomset languages,
such as defining a notion of counter-free higher dimensional automata.
More importantly, a central purpose of temporal logics is the specification and verification of program properties.
Thus, the obvious next steps of this work is to design efficient algorithms for deciding satisfiability and model-checking of \lsptl formulas.
It is already known (see \cite{AmraneBFF24}, Corollary 10) that these problems are decidable for \mso formulas, and thus also for \fo formulas.
We hope however that verifying \lsptl formulas directly should yield a better complexity than relying on their \fo translation.

%% file: sub/discussion.tex
\subsection{Event-based Pomset Temporal Logic}

In this section, we present a temporal logic that is weaker than \fo, which we call \defs{Event-based Pomset Temporal Logic} (\eptl).
This is a straightforward attempt at defining a local temporal logic for pomsets, with two ``Exists Next'' operators for the precedence order $<_P$ and for the event order $\eventorder_{P}$.
It is similar to the local temporal logic over Mazurkiewicz
traces defined in~\cite{DiekertG06local}.
This example showcases why it is not easy to design an \ltl-like logic over pomsets that is equivalent to \fo, and why we had to introduce the notion of sub-events.

\begin{definition}[Events-based Pomset Temporal Logic (\eptl)]\label{def-eptl}
    The syntax of \eptl formulas is defined by the following grammar, where $a \in \alphabet$:
    \[\form,\psi ::= a \mid \starter \mid \terminator \mid  \neg\form \mid \form \wedge \psi \mid \langle < \rangle\,\form \mid \ltlcirc \form \mid \form \until \psi\]
    An \eptl formula is evaluated over a single event $x$ of a pomset $P$:
    \begin{itemize}
        \item $P, x \models a$ if $\labelling_P(x) = a$,
        \qquad\qquad $P, x \models \starter$ if $x \in \sint{P}$,
        \qquad\qquad $P, x \models \terminator$ if $x \in \tint{P}$,
        \item $P, x \models \neg \form$ if $P, x \not\models \form$,
        \qquad\qquad \!$P, x \models \form \wedge \psi$ if $P,x\models \form$ and $P,x\models\psi$,
        \item $P, x \models \langle < \rangle \, \form$ if $\exists y \in P$,
         $y$ is a direct successor of $x$ for $<_P$
        and $P, y \models \form$,
        \item $P, x \models \ltlcirc \form$ if $\exists y \in P$, $y$ is a direct successor of $x$
        for ${\eventorder}_P$
        and $P, y \models \form$,
        \item $P, x \models \form \until \psi$ if $\exists y >_P x$
        $P, y \models \psi$ and $\forall z \in P$ such that
        $x \leq_P z <_P y$, $P, z \models \form$.
    \end{itemize}
\end{definition}

In order to define $P \models \form$, where $\form$ is an \eptl formula,
we need to fix a canonical $e \in P$ from which $\form$ will be interpreted.
This motivates \Cref{def-eptl-source}.

\begin{definition}
    \label{def-eptl-source}
    Given a non-empty pomset $P$, the \defs{source} of $P$, denoted $\src{P}$,
    is the minimal event according to $\eventorder$ of $\{e \in P \mid \forall f \in P, f \not< e\}$.
\end{definition}

We write $P \models \form$ if $P, \src{P} \models \form$.
    Note that the unary predicate verifying whether $x$ is a source is \fo-definable, define
    $\src{x} := \forall y. (\forall z. y \not< z) \Rightarrow x \eventorder y$.

\begin{proposition}
    \label{thm-exp-EPTL}
    \input{sub/prop/exp-EPTL.tex}
\end{proposition}

\begin{proof}[Proof (Sketch)]
First, \eptl is at most as expressive as \fo since its semantics expressed in \Cref{def-eptl} can be translated into first order formulas.
To show that the inclusion is strict, consider the following two families of pomsets, depicted in \Cref{fig-pn}:
\[
P_n = \pomset{a \interface}{a \interface}
	\left(
    \pomset{\interface a \pinterface}{\interface a \interface}
    \pomset{\interface a \interface}{\pinterface a \interface}
	\right)^n
	\pomset{\interface a}{\interface a}
\qquad\text{and}\qquad 
	P'_n =
	\pomset{a \interface}{a \interface}
	\pomset{\interface a \pinterface}{\interface a \interface}
	\pomset{\pinterface a \interface}{\interface a \interface}
	\left(
    \pomset{\interface a \pinterface}{\interface a \interface}
    \pomset{\interface a \interface}{\pinterface a \interface}
	\right)^n
	\pomset{\interface a}{\interface a}
\]
\begin{figure}[h]
    \centering
    \input{figures/eptl_pn.tex}
    \caption{Pomsets $P_2$ and $P'_2$}
    \label{fig-pn}
\end{figure}

\noindent
Consider the \fo formula
$\form = \exists x. \exists y. \exists z.\,
\src{x} \wedge x \rightarrow y \wedge x \eventorder^d z \eventorder^d y$,
where $x \rightarrow y$ and $x \eventorder^d y$ denote the direct successor w.r.t.\ $<$ and $\eventorder$, respectively.
The formula $\form$ separates the two languages, i.e., $P_n \models \form$
but $P'_n \not\models \form$ for all $n\in\mathbb{N}$.
Now we must show that any \eptl formula of size
$t$ cannot distinguish the pomsets $P_t$ and $P'_t$.
First, atomic formulas ($a \in \Sigma$, $\starter$ and $\terminator$) cannot distinguish between elements of $P_t$ and $P'_t$, since all events are labeled by $a$, and the interfaces are empty.
Moreover, notice that starting from $\src{P_t}$ (in orange in \Cref{fig-pn}) or $\src{P'_t}$ (in purple), there is a chain of length~$\geq t$ for both relations $\rightarrow$ and $\eventorder$.
Thus, using modalities $\langle < \rangle$ and $\ltlcirc$ to distinguish $\src{P_t}$ from $\src{P'_t}$
would require to reach the last elements of the pomset, which cannot be done
in less than $t$ steps.
Finally, the until modality cannot distinguish them either: this is a bit more technical but fairly standard.
So we conclude that the language $\langof{\form}$ is not definable in \eptl.
\end{proof}

\begin{remark}
To separate the two languages in \lsptl, one can remark that in the second slice of the pomsets $P_n'$, there is an event $a$ (in orange) which is both started and terminated.
Thus, the formula $\phi = \enext (\starter \wedge \terminator)$ is satisfied in all pomsets of the form $P_n'$, but in none of the pomsets $P_n$.
\end{remark}

%% file: sub/prop/exp-EPTL.tex
\eptl is stricly less expressive than \fo.

%% file: figures/eptl_pn.tex
\begin{subfigure}{.4\textwidth}
    \scalebox{0.8}{
        \begin{tikzpicture}%
            \def\hw{0.3}
            \def\th{0.5}
            \def\fh{0.1}
            \coordinate (O) at (0,0);
            \coordinate[above = 2cm of O] (Up);
            \coordinate[right = 5cm of O] (Right);
            \coordinate[above = 2cm of Right] (Up-right);
            \draw[-] (O) -- (Up);
            \draw[-] (Right) -- (Up-right);

            \filldraw[pomset-3] (0.2,\fh+3*\th) -- (1.8,\fh+3*\th) -- (1.8,\fh+3*\th+\hw) -- (0.2, \fh+3*\th+\hw) --  cycle; 
            \node at (1,\fh+3*\th+\hw*0.5) {$a$};

            \filldraw[pomset-1] (1.2,\fh+2*\th) -- (2.8,\fh+2*\th) -- (2.8,\fh+2*\th+\hw) -- (1.2, \fh+2*\th+\hw) --  cycle; 
            \node at (2,\fh+2*\th+\hw*0.5) {$a$};

            \filldraw[pomset-1] (2.2,\fh+1*\th) -- (3.8,\fh+1*\th) -- (3.8,\fh+1*\th+\hw) -- (2.2, \fh+1*\th+\hw) --  cycle; 
            \node at (3,\fh+1*\th+\hw*0.5) {$a$};

            \filldraw[pomset-1] (3.2,\fh) -- (4.8,\fh) -- (4.8,\fh+\hw) -- (3.2, \fh+\hw) --  cycle; 
            \node at (4,\fh+\hw*0.5) {$a$};
        \end{tikzpicture}
    }
\end{subfigure}
\qquad
\begin{subfigure}{.4\textwidth}
    \scalebox{0.8}{
        \begin{tikzpicture}%
            \def\hw{0.3}
            \def\th{0.5}
            \def\fh{0.1}
            \coordinate (O) at (0,0);
            \coordinate[above = 2cm of O] (Up);
            \coordinate[right = 6cm of O] (Right);
            \coordinate[above = 2cm of Right] (Up-right);
            \draw[-] (O) -- (Up);
            \draw[-] (Right) -- (Up-right);

            \filldraw[pomset-3] (1.2,\fh+3*\th) -- (2.8,\fh+3*\th) -- (2.8,\fh+3*\th+\hw) -- (1.2, \fh+3*\th+\hw) --  cycle; 
            \node at (2,\fh+3*\th+\hw*0.5) {$a$};

            \filldraw[pomset-2] (0.2,\fh+3*\th) -- (0.8,\fh+3*\th) -- (0.8,\fh+3*\th+\hw) -- (0.2, \fh+3*\th+\hw) --  cycle; 
            \node at (.5,\fh+3*\th+\hw*0.5) {$a$};

            \filldraw[pomset-1] (0.2,\fh+2*\th) -- (3.8,\fh+2*\th) -- (3.8,\fh+2*\th+\hw) -- (0.2, \fh+2*\th+\hw) --  cycle; 
            \node at (2,\fh+2*\th+\hw*0.5) {$a$};

            \filldraw[pomset-1] (3.2,\fh+1*\th) -- (4.8,\fh+1*\th) -- (4.8,\fh+1*\th+\hw) -- (3.2, \fh+1*\th+\hw) --  cycle; 
            \node at (4,\fh+1*\th+\hw*0.5) {$a$};

            \filldraw[pomset-1] (4.2,\fh) -- (5.8,\fh) -- (5.8,\fh+\hw) -- (4.2, \fh+\hw) --  cycle; 
            \node at (5,\fh+\hw*0.5) {$a$};
        \end{tikzpicture}
    }
\end{subfigure}

%% file: sub/appendix-sim-FO.tex
{
	\renewcommand{\thelemma}{\ref{lem-sim-FO}}
	\begin{lemma}
Fix $i,j \in \{1, \dots, k\}$. Then, the binary relation $(x, i) \relation (y,j)$
is definable by an $\fok{k}$ formula with two free variables $x$ and $y$.
	\end{lemma}
	\addtocounter{lemma}{-1}
}

\begin{proof}
    \input{sub/proofs/sim-FO-proof.tex}
\end{proof}

%% file: sub/proofs/sim-FO-proof.tex
We define a finite state automaton that scans the ST-sequence between the two positions~$x$ and~$y$.
During all the execution, the automaton keeps track the position of the event represented by the $i$-th position of~$x$.
Showing that this automaton is counter-free yields an $\fok{k}$ formula for the $\relation$ relation, thus proving the lemma.
Formally, given $a \in \alphabet$ and $i,j \in \{1, \dots, k\}$,
we define a finite state automaton $\automata_{i,j,a} = (\stset{k}, Q, q_0, F, \delta)$
over the alphabet $\stset{k}$, parameterized by $i, j, a$, as follows:
\begin{itemize}
    \item The set of states $Q$ is given by:
    \[
    Q = \{\bot, \top, \sqcup\} \cup (\concb{k} \times \{1, \dots, k\})
    \]
    The states $\bot, \top, \sqcup$ represent the initial state, final state, and sink state, respectively.
    All other states are of the form $(C, \ell)$, where $C$ is the list of currently active events,
    and the event $a$ that we are following is the $\ell$-th element in this conclist.
    \item The initial state is $q_0 = \bot$.
    \item The set of accepting states is $F = \{\top\} \cup \{(q,j) \mid q \in Q\}$ (recall that $j$ is a parameter fixed in advance).
    \item The transition function $\delta$ is defined as follows, for any $P \in \stset{k}$:%
        \begin{align*}
            \delta(\bot, \stelement) & = (\tint{\stelement}, i) && \text{if $\tint{\stelement}[i]$ is a non-terminating $a$}\\
            \delta(\bot, \stelement) & = \top && \text{if $\tint{P}[i]$ is a terminating $a$ and $i = j$}\\
            \delta((\sint{\stelement}, \ell), \stelement) & = (\tint{\stelement}, \ell')&& \text{for any $\ell,\ell'$ when $\sint{\stelement}[\ell] = \tint{\stelement}[\ell']$ and is an $a$}\\
            \delta((\sint{\stelement}, j), \stelement) & = \top && \text{if $\sint{\stelement}[j]$ is a terminating $a$}\\
            \delta(q, \stelement) &= \sqcup \text{ (the sink state) } && \text{otherwise}
        \end{align*}
        where $C[i]$ denotes is the $i$-th element of the conclist $C$.
\end{itemize}
\begin{figure}[htbp]
    \centering
    \input{figures/aut_sim.tex}
    \caption{$\automata_{1,2,a}$, sink state and identities not drawn.}
    \label{fig-automata-sim-proof}
\end{figure}

\Cref{fig-automata-sim-proof} shows the automaton $\automata_{1,2,a}$, for $\alphabet = \{a,b\}$ and $k = 2$.
For the sake of readability, when an interface is represented enclosed in parenthesis (e.g. $[(\ibullet) a]$), we mean
that both transitions exist, with and without the interface (e.g. $[a]$ and $[\ibullet a]$).
The event $a$ that we are tracking is colored in red.

Note that $\automata_{i,j,a}$ is deterministic and complete. We now show that it is counter-free.
So we need to find an $n \in \mathbb{N}$ such that, for any non-empty word ${w \in (\stset{k})^\ast}$ and any state~$q$, we have $\delta(q, w^n) = \delta(q, w^{n+1})$.
Choose $n = k + 2$. Consider an arbitrary state $q$ and a non-empty word $w$.
We have three cases:
\begin{enumerate}
    \item\label{item-lem-sim-FO-sink} If $q = \top$ or $q = \sqcup$, the only accessible state on a non-empty word is the sink state $\sqcup$, verifying the property.
    \item If $q = (c, \ell) \in \concb{k} \times \{1, \dots, k\}$, let us consider $\delta((c,\ell), w)$. 
    \begin{enumerate}
        \item If $\delta((c,\ell),w) = (c,\ell)$, we indeed have $\delta((c,\ell),w^{n+1}) = \delta((c,\ell),w^{n}) = (c,\ell)$. 
        \item If $\delta((c,\ell),w) = (c', \ell')$, with $c' \neq c$, then $\delta((c', \ell'), w) = \sqcup$ since either
        the starting interface or the terminating interface of $w$ will be incompatible with $c'$. Thus, the property is verified. %
        \item If $\delta((c,\ell),w) = (c, \ell')$, with $\ell \neq \ell'$, assume that $\ell > \ell'$ (the opposite case is similar).
        We claim that, as we keep iterating the word $w$, the index $\ell$ will keep decreasing strictly, until we reach either the state $\top$ or $\sqcup$ after at most $k$ iterations. So at iteration $k+1$, we end up in $\sqcup$, where we stay, making the property true.
        To prove our claim, let us write $(c, \ell_i) = \delta((c,\ell),w^i)$ the sequence of visited states. We need to show that the sequence $(\ell_i)$ is strictly decreasing.

        The example below illustrates what might happen, with a pomset of dimension $k=3$. We start in state $(aaa,3)$, so three $a$'s are running concurrently, and we are tracking the third one (in red).
		When the automaton reads the word $w$ below, the first $a$ is terminated, and another $a$ is started, but this new $a$ happens to be placed after the other two according to event order.
		The tracked $a$ ends up in position 2, so the state of the automaton is now $(aaa,2)$.
        \[
          w = \left[\begin{array}{ccc}
          \ibullet a \pinterface \\ \color{blue} \ibullet a \ibullet \\ \color{red} \ibullet a \ibullet
          \end{array}\right]
    	  \left[\begin{array}{ccc}
    		\color{blue} \ibullet a \ibullet \\ \color{red} \ibullet a \ibullet \\ \pinterface a \ibullet
          \end{array}\right]
      	\qquad\qquad
		  \begin{array}{l}
      	  \delta((aaa, 3), w) = (aaa, 2)\\[6pt]
		  \delta((aaa, 3), w^2) = (aaa, 1)
		  \end{array}
         \]
        Let us assume that for some~$i$, $\ell_i > \ell_{i+1}$, we will show that at the next step, $\ell_{i+1} > \ell_{i+2}$.
        Let $P_w \in \iipoms$ be the pomset generated by gluing~$w$ (the automaton ensures that $w$ is coherent, otherwise we end up in the sink state $\sqcup$).
        Let $p_i, p_{i+1} \in P_w$ denote the events of $P_w$ at position $\ell_i$ and $\ell_{i+1}$ respectively, in the starting interface $\sint{P_w}$. (In the example above, they are the red and blue events, respectively.)
		Since $\ell_i > \ell_{i+1}$, we must have $p_{i+1} \eventorder p_i$ according to event order.
		In the terminating interface $\tint{P_w}$, both events are still active. Event $p_i$ ends up at position $\ell_{i+1}$ and $p_{i+1}$ ends up at position $\ell_{i+2}$.
		Since these are still the  same events, the event order did not change, so we still have $p_{i+1} \eventorder p_i$, i.e., $\ell_{i+1} > \ell_{i+2}$ as required.
    \end{enumerate}
    \item If $q = \bot$, we have $\delta(\bot, w) \neq \bot$ since $w$ is non-empty and there is no incoming transition in $\bot$.
    So we land in one in the previous cases, with one extra step.
    Hence, $\delta(\bot, w^{k+3}) = \delta(\bot, w^{k+2})$.
\end{enumerate}
This proves that $\automata_{i,j,a}$ is counter-free. By \Cref{prop-aperiodic-autom}, let
$\theta_{i,j,a}$ be a closed first-order formula equivalent that recognizes the same language as $\automata_{i,j,a}$.
By $\theta_{i,j,a}^{x \rightarrow y}$, we denote the restriction of $\theta_{i,j,a}$ to the portion
of the ST-sequence located between $x$ and $y$, {i.e.}\ where subformulas of the form $\exists z.\form$ are inductively
replaced by $\exists z. x \leq z \wedge z \leq y \wedge \form$.
Then the following $\fok{k}$ formula defines ${(x,i) \relation (y,j)}$:
\[(x,i) \relation (y,j) := \bigvee_{a \in \alphabet} (x < y \wedge \theta_{i,j,a}^{x \rightarrow y})
                                            \vee (y < x \wedge \theta_{j,i,a}^{y \rightarrow x})
\qedhere\]

%% file: sub/appendix-FOP.tex
{
	\renewcommand{\thetheorem}{\ref{thm-FOP-to-FOST}}
	\begin{theorem}
    \input{sub/prop/FOP-to-FOST.tex}
	\end{theorem}
	\addtocounter{theorem}{-1}
}

\begin{proof}
    \input{sub/proofs/FOP-to-FOST-proof.tex}

\end{proof}

\begin{lemma}
	Withed fixed dimension $k$, the worst-case size of the translated formula is $|\translated\form| = \Theta(k^{|\form|})$.
	More precisely, it is exponential in the number of quantifiers.
\end{lemma}

\begin{proof}
	We have that $|\translated{\exists x_1. \exists x_2. \dots \exists x_n. a(x_1)}|
	= |\coh{k}| + k^n \times O(1)$, where the original size is
	$|\exists x_1. \exists x_2. \dots \exists x_n. a(x_1)| = n + 1$.
	Therefore, we have that the worst-case size of the translated formula is
	$|\translated{\form} | = \Omega(k^{|\form|})$.
	The fact that $|\translated\form| = O(2^{|\form|})$ is proven by induction over
	$\form$. Remark that $|\coh{k}| = O(1)$.
	\begin{itemize}
		\item $|\translated{\neg\form}| = 1 + O(k^{|\form|}) = O(k^{|\neg\form| - 1}) = O(k^{|\neg\form|})$,
		\item $|\translated{\form_1 \wedge \form_2}| = 1 + O(k^{|\form_1|})
		+ O(k^{|\form_2|}) = O(k^{|\form_1 \wedge \form_2|-1}) = O(k^{|\form_1 \wedge \form_2|})$,
		\item $|\translated{\exists x.\form}| = k \times O(k^{|\form|}) + k = O(k^{|\form| + 1})
		= O(k^{|\exists x. \form|})$,
		\item $|\translated{a(x)}| = |\translated{\starter(x)}| = |\translated{\terminator(x)}|
		= |\translated{x < y}| = |\translated{x \eventorder y}| = O(1)$ since $k$ is fixed.
	\end{itemize}
\end{proof}

%% file: sub/proofs/FOP-to-FOST-proof.tex
The translation $\translated{\form}$ of $\form$ must exactly accept the words $w = P_1 P_2 \cdots P_n \in \stset{k}^*$ that verifies two properties.
First, $w$ is a coherent ST-sequence. Second, if $w$ is a coherent ST-sequence, then the pomset $\gluew{w} = P_1 * P_2 * \cdots * P_n$ validates~$\form$.
The first property is local and can be checked by the following first-order formula, where $x \rightarrow y := x < y \wedge \neg \exists z.\; x < z < y$.
\[ \coh{k} =
        \forall x. \forall y.\; x \rightarrow y \Rightarrow
        \bigvee_{\underset{P_1 \cdot P_2 \text{ is coherent}}{P_1, P_2 \in \stset{k}}}
        P_1(x) \wedge P_2(y)
\]
For the second property, we define by induction a translation $\translated{\psi}_\tau$, where $\psi$ is a subformula of $\form$.
This translation is indexed by a function $\tau: \calV(\psi) \rightarrow \{1, \dots, k\}$, where $\calV(\psi)$ denotes the set of the free variables of $\psi$.
As explained in \Cref{sec-sim-relation}, the pair $(x, \tau(x))$ is used to keep track of events, where $\tau(x)$ is a position within the starter/terminator designated by~$x$.
The invariant is as follows: if $w$ is a coherent ST-sequence, then $w, \interp \models \translated{\psi}_\tau$ if and only if
$\gluew{w}, \interp' \models \psi$, where $\interp'(x)$ is the $\tau(x)$-th event of $\interp(x)$. We now proceed to the induction.
\begin{itemize}
    \item If $\psi = \neg \theta$, then $\translated{\psi}_\tau = \neg \translated{\theta}_{\tau}$,
    \item If $\psi = \theta_1 \wedge \theta_2$, then $\translated{\psi}_\tau = \translated{\theta_1}_{\tau} \wedge \translated{\theta_2}_{\tau}$,
    \item If $\psi = \exists x. \theta$, then $\translated{\psi}_\tau = \bigvee_{i\in \{1, \dots, k\}} \exists x.\;\translated{\theta}_{\tau[x \mapsto i]}$.
    Intuitively, when we existentially quantify a first-order variable over pomsets (where it represents an event),
    we need to guess a starter/terminator $x$, and a position~$i$. Since there are only
    $k$ possible positions, we use a finite disjunction,
    \item If $\psi = a(x)$, then $\translated{\psi}_{[x \mapsto i]}$ is the disjunction of all $P(x)$, where $P \in \stset{k}$ is
    such that the $i$-th element of $P$ is an $a$, %
    \item If $\psi = \starter(x)$, then
    $\translated{\psi}_{[x \mapsto i]} = \bigvee_{j \in \{1, \dots, k\}} \forall y.\; (x,i) \relation (y,j) \wedge \starter(y,j)$,
    where $\starter(y,j)$ is the disjunction of all $P(x)$, where the $j$-th element of $P \in \stset{k}$ is in the
    starting interface,
    \item If $\psi = \terminator(x)$, the construction is similar,
    \item If $\psi = x < y$, then
    \[\translated{\psi}_{[x \mapsto i, y \mapsto j]} = \bigwedge_{i',j' \in \{1, \dots, k\}} \forall x'. \forall y'.\;
    ((x,i) \relation (x',i') \wedge (y,j) \relation (y',j')) \Rightarrow x' < y'\]
    \item Finally, if $\psi = x \eventorder y$, then
    \[\translated{\psi}_{[x \mapsto i, y \mapsto j]} = \bigvee_{1 \leq i' < j' \leq k} \exists z.\; (x,i) \relation (z,i') \wedge
    (y,j) \relation (z, j')\]%
\end{itemize}

To conclude, we write the translated of $\form$ as $\translated{\form} = \coh{k} \wedge \bigvee_{\tau: \calV(\form) \rightarrow \{1, \dots, k\}}\translated{\form}_\tau$.

%% file: sub/appendix-ltlp-sube.tex
{
    \renewcommand{\theproposition}{\ref{prop-ltlp-sube}}
    \begin{proposition}
        \input{sub/prop/ltlp-sube}
    \end{proposition}
    \addtocounter{proposition}{-1}
}
\begin{proof}
    \input{sub/proofs/ltlp-sube-proof.tex}
\end{proof}

%% file: sub/proofs/ltlp-sube-proof.tex
First observe that $\sim^s$ is an equivalence relation over events. Indeed, it is reflexive since $m <^s q$ implies $m \neq q$ and
symmetric by construction. As for transitivity, assume that $m \sim^s q \sim^s r$. To show that $m \sim^s r$, it is enough to show $m \not<^s r$.
Assume by contradiction that there is $t \in P$ s.t.\ $t\parallel m$, and $t < r$. Then, we prove $t \parallel q$ as follows:
\begin{itemize}
    \item If $t < q$, then $m <^s q$ because $t \parallel m$. But this is false since $m \sim^s q$.
    \item If $q < t$, then $q < r$ since $t < r$. This implies that $q <^s r$, which is false by hypothesis.
\end{itemize}
Since $q \sim^s r$, $q \not<^s r$. Therefore, as $t\parallel q$, it cannot be that $t < r$. Since this was a hypothesis, we have a contradiction, and
$m \sim^s r$.

We can now proceed to the main proof. It is sufficient to prove that $P, (x,m) \models \form$ implies $P, (y,q) \models \form$.
First, observe that $x = y$ by definition of $\equiv$, hence we write $(x,q)$ from now on.
We proceed by induction over $\form$.
\begin{itemize}
    \item $P, (x,m) \models a$ implies $\labelling(x) = a$, thus $P,(x,q) \models a$
    \item $P, (x,m) \models \starter$ implies $x \sim^s m$. Since $m \sim^s q$, this yields $x \sim^s q$, which in turn becomes $P, (x,q) \models \starter$
    \item $P, (x,m) \models \terminator$ implies that there is no $r \in P$ such that $x \parallel r$ and $m <^s r$. Since $m \sim^s q$,
    $m <^s r$ is equivalent to $q <^s r$. This in turn implies $P, (x,q) \models \terminator$.
    \item $P, (x,m) \models \neg\psi$ implies $P, (x,m) \not\models \psi$, which implies, by induction hypothesis,
    $P, (x,q) \not\models \psi$. Thus, $P, (x,q) \models \neg \psi$.
    \item $P, (x,m) \models \psi_1 \wedge \psi_2$ implies $P, (x,m) \models \psi_1$ and $P, (x,m) \models \psi_2$. Hence, by
    induction hypothesis, $P, (x,q) \models \psi_1$ and $P, (x,q) \models \psi_2$. It comes that $P, (x,q) \models \psi_1 \wedge \psi_2$.
    \item $P, (x,m) \models \enext \psi$ implies that there exists a sub-event $(y,r)$ such that $(x,m) \prec_1 (y,r)$ and $P, (y,r) \models \psi$.
    Because $(x,m) \prec_1 (y,r)$, we have that $m <^s r$ and there is no $s \in P$ such that $m <^s s <^s r$. Since $m \sim^s q$,
    it comes that $q <^s r$ and there is no $s$ such that $q <^s s <^s r$. Thus, $(x,q) \prec_1 (y,r)$ and $P,(x,q) \models \enext \psi$.
    \item $P, (x,m) \models \ltlcirc \psi$ implies that there is $y$ such that $x \sim^s y$, $y$ is a direct successor of $x$ by $\eventorder$
    and $P, (y,m) \models \psi$.
    Since $m \sim^s q$, $(y,q)$ is a sub-event and $P, (y,q) \models \psi$ by induction. Hence, $P, (x,q) \models \ltlcirc\psi$.
    \item For $\ltlrcirc \psi$, the argument is identical.
    \item Finally, for $P, (x,m) \models \psi_1 \until \psi_2$, fix $(y,r)$ to be such that $P,(y,r) \models \psi_2$ and $(x,m) \preccurlyeq (y,r)$ and any
    $(z,s)$ such that $(x,m) \preccurlyeq (z,s) \prec (y,r)$ verifies $P, (x,s)\models \psi_1$.
    Since $m \sim^s q$, we have that $(x,m)\preccurlyeq (y,r)$ implies $(x,q) \preccurlyeq (y,r)$, and the same goes for the
    intermediary $(z,s)$. Hence, $P, (x,q) \models \psi_1 \until \psi_2$. \qedhere
\end{itemize}

%% file: sub/appendix-ltlp-induction.tex
To prove \Cref{prop-induction}, we first introduce \Cref{def-evensparse} and establish \Cref{lem-order-start,lem-prec-sube,lem-lsptl-next,lem-induction}.

\begin{definition}
    \label{def-evensparse}
    Given an \evensparse ST-decomposition $P_1 P_2 \cdots P_{2n}$ and a sub-event $(x,m)$, we say that $(x,m)$ \defs{belongs} to the pair $P_{2i-1} P_{2i}$
    if $m$ is started by $P_{2i-1}$.  
\end{definition}

\begin{lemma}
    \label{lem-order-start}
    Let $P$ be a pomset with \evensparse ST-decomposition $P_1 P_2 \cdots P_{2n}$. Let $x,y\in P$ such that $x$ starts in $P_{2i-1}$ and
    $y$ starts in $P_{2j-1}$. Then, $i < j$ if and only if $x <^s y$.
\end{lemma}
\begin{proof}
    Assume that $i < j$. Then, fix $z$ terminating in $P_{2i}$. We have $x \parallel z$ since they are both running in $P_{2i}$
    and $z < y$ since $y$ starts in $P_{2j-1}$. Therefore, $x <^s y$.
    Conversely, assume that $x <^s y$. Fix $z$ such that $x \parallel z$ and $z < y$. We must have that $z$ is not yet terminated in $P_{2i-1}$ since $x$ starts here,
    but it must terminate before $P_{2j-1}$. Therefore, $i < j$.
\end{proof}

\begin{lemma}\label{lem-prec-sube}
    Let $P$ be a pomset with \evensparse decomposition $P_1 P_2 \cdots P_{2n}$.
    Fix two sub-events: $(x,m)$ belonging to $P_{2i-1}P_{2i}$ and $(y,q)$ belonging to $P_{2j-1}P_{2j}$.
    If $(x,m)\prec(y,q)$, then $i < j$.
\end{lemma}
\begin{proof}
    If $x < y$, then $i<j$ because $x$ must be terminated before $y$ starts.
    Otherwise, $x = y$ and $m <^s q$. By \Cref{lem-order-start}, we have $i < j$ since $m$ starts in $P_{2i-1}$ and $q$ starts in $P_{2j-1}$.
\end{proof}

\begin{lemma}
    \label{lem-lsptl-next}
    Let $P$ be a pomset with \evensparse decomposition $P_1 P_2 \cdots P_{2n}$.
    Given a sub-event $(x,m)$ of $P_{2i-1}P_{2i}$,
    all sub-events $(y,q)$ such that $(x,m) \prec_1 (y,q)$ belong to $P_{2i+1}P_{2i+2}$.
\end{lemma}
\begin{proof}
    By \Cref{lem-prec-sube}, $m$ starts before $q$. Further, if $q$ is not yet started in $P_{2i+1}$, then any event $r$ started in $P_{2i+1}$
    would be such that $m <^s r <^s q$, which is impossible since $(x,m) \prec_1 (y,q)$.
\end{proof}

The next \Cref{lem-induction} is the main technical part of the proof

\begin{lemma}
    \label{lem-induction}
    Given an LTL-ST formula $\form$, there exists $\translated{\form}_\starter$ and
    $\translated{\form}_\terminator$ such that for any pomset $P$ of \evensparse ST-decomposition $P_1 P_2 \cdots P_{2n}$,
    for any $i \in \{1, \dots, n\}$ for any sub-event $(x,m)$ of $P_{2i-1} P_{2i}$,
    $P_{2i-1} P_{2i} \cdots P_{2n} \models \form$ if and only if $P, (x,m) \models \translated{\form}_\starter$ and
    $P_{2i} P_{2i+1}  \cdots P_{2n} \models \form$ if and only if $P, (x,m) \models \translated{\form}_\terminator$.
\end{lemma}
\begin{proof}
    \input{sub/proofs/ltlp-induction-proof.tex}
\end{proof}

We can now prove \Cref{prop-induction}. Let us first recall its statement.

{
    \renewcommand{\theproposition}{\ref{prop-induction}}
    \begin{proposition}
        \input{sub/prop/ltlp-induction}
    \end{proposition}
    \addtocounter{proposition}{-1}
}

\begin{proof}
    Fix $\translated{\form} = \translated{\form}_\starter$ as defined by \Cref{lem-induction}.
	Recall that $P \models \translated{\phi}$ means that $P, \src{P} \models \translated{\phi}$, where $\src{P} = (x,x)$ for $x$ an $<$-minimal event of~$P$.
	Since $x$ is $<$-minimal, it must be started by $P_1$, or else any event $y$ terminated
    by $P_2$ would be such that $y<x$. Therefore, $(x,x)$ belongs to $P_1P_2$, hence the result.
\end{proof}

\begin{lemma}
    The size of $\translated{\form}$ in the worst case is exponential in the size of $\form$.
\end{lemma}

\begin{proof}
    First, let us show that $|\translated{\form}| = \Omega(2^{|\form|})$ in the worst case.
    Fix $\form_0 = P$ for some starter or terminator $P$ and $\form_{n+1} = \form_n \until P$.
    Then, $|\translated{\form_0}|$ is of constant size.
    Now, $|\translated{\form_{n+1}}| = |\translated{\form_{n+1}}_\starter| =
    |\translated{\form_n \until P}_\starter| = |(\translated{\form_n}_\starter \wedge \translated{\form_n}_\terminator)
    \until
    (\translated{P}_\starter \vee (\translated{\form_n}_\starter \wedge \translated{P}_\terminator))| > 2 \cdot |\translated{\form_n}|$.
    Hence, $\translated{\form_{n}} = \Omega(2^n) = \Omega(2^{|\form|})$

    Let us now proceed to show that $|\translated{\form}|$ at most exponential in $|\form|$ in any case. We prove this by induction over $\form$.
    \begin{itemize}
        \item If $\form$ is an atomic formula $P$, then $\translated{\form}$ is of constant size.
        \item If $\form$ is a boolean combination of other formulas, $\translated{\form}$ is of polynomial size in the sizes of the translations of those formulas since $\neg$ and $\wedge$
        are directly translated.
        \item If $\form = \ltlnext \psi$, then $\translated{\form}$ is in polynomial size in $\translated{\psi}$, which is exponential in $|\psi| = |\varphi| - 1$.
        \item If $\form = \psi_1 \until \psi_2$, then $|\translated{\form}_\starter| = C + 2\cdot|\translated{\psi_1}_\starter| + |\translated{\psi_1}_\terminator| + |\translated{\psi_2}_\starter| + |\translated{\psi_2}_\terminator|$,
        This is at most exponential in the size of $\psi_1 \until \psi_2$ because each of the translated subformulas are, and $|\translated{\form}|$ is a polynomial in the sizes of the translations of these subformulas.
    \end{itemize}
\end{proof}

%% file: sub/proofs/ltlp-induction-proof.tex
We prove this by induction over $\form$.
\begin{itemize}
    \item For the base case where $\form$ is a starter or terminator $P^* \in \stset{k}$, denote $a_1, a_2, \cdots, a_\ell$, the labels of the events
    in $P^*$, ordered by $\eventorder$. If $P^*$ is not a starter, then define $\translated{P^*}_\starter = \bot$.
    Otherwise, fix
    $\translated{P^*}_\starter = \ltlpar (\neg\ltlrcirc \top \wedge \bigwedge_{i=1}^\ell \ltlcirc^{i-1}\form_i \wedge \neg\ltlcirc^{\ell}\top)$ where $\form_i = a_i \wedge \starter$ if
    $P^*$ starts its $i$-th event, or $\form_i = a_i \wedge \neg\starter$ otherwise. 
    $\translated{P^*}_\terminator$ is defined similarly.
    \item $\translated{\neg\psi}_\starter = \neg\translated{\psi}_\starter$
    \item $\translated{\neg\psi}_\terminator = \neg\translated{\psi}_\terminator$
    \item $\translated{\psi_1 \wedge \psi_2}_\starter = \translated{\psi_1}_\starter \wedge \translated{\psi_2}_\starter$
    \item $\translated{\psi_1 \wedge \psi_2}_\terminator = \translated{\psi_1}_\terminator \wedge \translated{\psi_2}_\terminator$
    \item $\translated{\ltlnext \psi}_\starter = \translated{\psi}_\terminator$. Indeed, $P_{2i-1} P_{2i} \cdots P_{2n}\models \ltlnext \psi$
    if and only if $P_{2i} P_{2i+1} \cdots P_{2n}\models\psi$, which is in turn equivalent to $P, (x,m) \models \translated{\psi}_\terminator$ by
    induction hypothesis.
    \item $\translated{\ltlnext \psi}_\terminator
    = \enext \translated{\psi}_\starter$
    \begin{itemize}
        \item Assume that $P_{2i} P_{2i+1} \cdots P_{2n} \models \ltlnext\psi$. This is true if and only if $P_{2i+1} P_{2i+2} \cdots P_{2n} \models \psi$,
        which is equivalent by induction hypothesis to the fact that any sub-event $(y,q)$ of $P_{2i+1}P_{2i+2}$ satisfies $P, (y,q) \models \translated{\psi}_\starter$. However the sub-events accessible by
        $\enext$ from $(x,m)$ are all sub-events of $P_{2i+1}P_{2i+2}$ by \Cref{lem-lsptl-next} and the definition of $\enext$ for \lsptl.
        Therefore, we have that $P, (x,m) \models \enext \translated{\psi}_\starter$.
        \item Conversely, if $P, (x,m) \models \enext \translated{\psi}_\starter$, then there exists a $(y,q)$ belonging
        to $P_{2i+1} P_{2i+2}$ verifying $P, (y,q) \models \translated{\psi}_\starter$. This implies that $P_{2i+1} P_{2i+2} \cdots P_{2n} \models \psi$, which yields the desired conclusion.
    \end{itemize}
    \item $\translated{\psi_1 \until \psi_2}_\starter = (\translated{\psi_1}_\starter \wedge \translated{\psi_1}_\terminator)
    \until
    (\translated{\psi_2}_\starter \vee (\translated{\psi_1}_\starter \wedge \translated{\psi_2}_\terminator))$
    \begin{itemize}
        \item Assume that $P_{2i-1} \cdots P_n \models \psi_1 \until \psi_2$.
        If $P_{2i-1} \cdots P_{2n} \models \psi_2$ then $P, (x,m)\models \translated{\psi_2}_\starter$
        by induction hypothesis, which implies that $P, (x,m) \models
        (\translated{\psi_1}_\starter \wedge \translated{\psi_1}_\terminator)
        \until
        (\translated{\psi_2}_\starter \vee (\translated{\psi_1}_\starter \wedge \translated{\psi_2}_\terminator))$.

        Now, assume that $P_{2i-1} \cdots P_{2n} \not\models \psi_2$.
        Since $P_{2i-1} \cdots P_n \models \psi_1 \until \psi_2$,
        fix $j > i$ such that $P_{2j-1} \cdots P_n \models \psi_2$ or
        $P_{2j} \cdots P_n \models \psi_2$ and all intermediate $\ell$'s verify
        $P_\ell \cdots P_n \models \psi_1$. Then, fix $(y,q)$, sub-event in 
        $P_{2j-1}P_{2j}$ such that $(x,m) \prec (y,q)$ (there exists such a sub-event:
        if $x$ is still running in $P_{2j-1}$, take $(x, q)$ where $q$ is an event started
        in $P_{2j-1}$; if $x$ is not running anymore, take $(q,q)$ with the same $q$).
        Then, we have that $P, (y,q)\models \translated{\psi_2}_\starter \vee
        (\translated{\psi_1}_\starter \wedge \translated{\psi_2}_\terminator)$ by induction.
        Now, fix $(z,r)$ such that $(x,m) \preccurlyeq (z,r) \prec (y,q)$.
        This implies that $(z,r)$ belongs to one of $P_{2\ell-1}P_{2\ell}$ with
        $i \leq \ell < j$. Therefore, we have that $P, (z,r) \models \translated{\psi_1}_\starter \wedge \translated{\psi_1}_\terminator$
        by induction.
        Hence, we have that $P, (x,m) \models
        (\translated{\psi_1}_\starter \wedge \translated{\psi_1}_\terminator)
        \until
        (\translated{\psi_2}_\starter \vee (\translated{\psi_1}_\starter \wedge \translated{\psi_2}_\terminator))$
        \item Conversely, assume that $P, (x,m) \models
        (\translated{\psi_1}_\starter \wedge \translated{\psi_1}_\terminator)
        \until
        (\translated{\psi_2}_\starter \vee (\translated{\psi_1}_\starter \wedge \translated{\psi_2}_\terminator))$.
        Therefore, fix $(y,q)$ such that $(x,m) \preccurlyeq (y,q)$, $P, (y,q) \models 
        \translated{\psi_2}_\starter \vee (\translated{\psi_1}_\starter \wedge \translated{\psi_2}_\terminator)$
        and, for any $(z,r)$ such that $(x,m) \preccurlyeq (z,r) \prec (y,q)$, $P,(z,r) \models  \translated{\psi_1}_\starter \wedge \translated{\psi_1}_\terminator$.
        \begin{itemize}
            \item If $(x,m) = (y,q)$, then $P,(x,m) \models \translated{\psi_2}_\starter \vee (\translated{\psi_1}_\starter \wedge \translated{\psi_2}_\terminator)$.
            Fix $i$ such that $(x,m)$ is a sub-event of $P_{2i-1}P_{2i}$. By induction hypothesis, we have that
            $P_{2i-1} \cdots P_{2n} \models \psi_1 \until \psi_2$ (by satisfying $\psi_2$ in zero or one step).
            \item If $(x,m) \prec (y,q)$, then fix $i,j$ such that $(x,m)$ is a sub-event of $P_{2i-1} P_{2i}$ and
            $(y,q)$ is a sub-event of $P_{2j-1} P_{2j}$. Since $(x,m) \prec (y,q)$ we have that $i < j$ by \Cref{lem-prec-sube}.
            We have that $P_{2i-1} \cdots P_{2n} \models \psi_1$ and $P_{2i} \cdots P_{2n} \models \psi_1$
            since $(x,m) \models \translated{\psi_1}_\starter \wedge \translated{\psi_1}_\terminator$.
            Fix now $\ell$ such that $i < \ell < j$. if $x$ is still running in $P_{2\ell-1}$, fix $z = x$. Elsewhere, fix $z$ to be an event terminated
            by $P_{2\ell}$. Fix $r$ as an event started in $P_{2\ell - 1}$. We have that $(x,m) \prec (z,r) \prec (y,q)$.
            Now, observe that $2i-1 < 2\ell-1 < 2\ell < 2j-1$. Therefore, $P_{2\ell-1} \cdots P_{2n} \models \psi_1$
            and $P_{2\ell} \cdots P_{2n} \models \psi_1$.
            Finally, we have that $P_{2j-1} \cdots P_{2n} \models \psi_1 \until \psi_2$ (in zero or one step) since
            $P,(y,q) \models \translated{\psi_2}_\starter \vee (\translated{\psi_1}_\starter \wedge \translated{\psi_2}_\terminator)$
            Therefore, we can conclude that $P_{2i-1} \cdots P_{2n} \models \psi_1 \until \psi_2$.
        \end{itemize}
    \end{itemize}
    \item The translation $\translated{\psi_1 \until \psi_2}_\terminator$ can be inferred from the previous cases, since $\psi_1 \until \psi_2$ is equivalent to $\psi_2 \vee (\psi_1 \wedge \ltlnext (\psi_1 \until \psi_2))$.
	This yields:
	\begin{align*}
	\translated{\psi_1 \until \psi_2}_\terminator
	&= \translated{\psi_2 \vee (\psi_1 \wedge \ltlnext (\psi_1 \until \psi_2))}_\terminator\\
	&= \translated{\psi_2}_\terminator \vee (\translated{\psi_1}_\terminator \wedge \enext \translated{\psi_1 \until \psi_2}_\starter)\\
	&= \translated{\psi_2}_\terminator \vee (\translated{\psi_1}_\terminator \wedge \enext ((\translated{\psi_1}_\starter \wedge \translated{\psi_1}_\terminator)
    \until
    (\translated{\psi_2}_\starter \vee (\translated{\psi_1}_\starter \wedge \translated{\psi_2}_\terminator))))
	\end{align*}
\end{itemize}
This concludes the induction.

%% file: sub/appendix-ltlp-to-fop.tex
{
    \renewcommand{\theproposition}{\ref{lem-LTLP-to-FOP}}
    \begin{proposition}
        \input{sub/prop/ltlp-to-fop.tex}
    \end{proposition}
    \addtocounter{proposition}{-1}
}
\begin{proof}
    \input{sub/proofs/ltlp-to-fop-proof.tex}
\end{proof}

\begin{lemma}
    The size of $\translated{\form}$ in \Cref{lem-LTLP-to-FOP} is linear in the size of $\form$.
\end{lemma}

\begin{proof}
    We prove this by induction over $\form$ :
    \begin{itemize}
        \item If $\form$ is an atomic formula $a \in \Sigma$, $\starter$ or $\terminator$, then the size of $\translated{\form}$ is constant
        \item If $\form = \neg\psi$, then $\translated{\form} = \neg\translated{\psi}$ is of size linear in $\neg\psi$ since $\translated{\psi}$ is of size linear in $\psi$ by induction hypothesis
        \item If $\form = \psi_1 \wedge \psi_2$, then the size of $\translated{\form} = \translated{\psi_1} \wedge \translated{\psi_2}$ is linear in $\psi_1 \wedge \psi_2$ by induction hypothesis
        \item If $\form = \enext \psi$, then $|\translated{\form}| = C + |\translated{\psi}|$ which is linear in $|\form|$ by induction hypothesis.
        \item The same goes for $\form = \ltlcirc \psi$ or $\ltlrcirc \psi$.
        \item Finally, if $\form = \psi_1 \until \psi_2$, then $|\translated{\form}| = C + |\translated{\psi_1}| + \translated{\psi_2}|$, which is still linear in $|\translated{\form}| = |\translated{\psi_1}| + |\translated{\psi_2}|$.
    \end{itemize}
\end{proof}

%% file: sub/proofs/ltlp-to-fop-proof.tex
The essential idea is that the semantics of \lsptl, expressed in natural language in \Cref{def-ltlp}, can be written in first-order.
Let us first define a few useful \fo formulas.
\begin{itemize}
    \item $x \parallel y := \neg x < y \wedge \neg y < x$
    \item $x <^s y := \exists z. x \parallel z \wedge z < y$
    \item $x \sim^s y := \neg x <^s y \wedge \neg y <^s x$
    \item $x \lesssim^s y := x <^s y \vee x \sim^s y$
    \item $\fose(x,m) := x \parallel m  \wedge x \lesssim^s m$ (\ie $(x,m)$ is a sub-event)
    \item $(x,m) \prec (y,q) := x < y \vee (x = y \wedge m <^s q)$
    \item $(x,m) \preccurlyeq (y,q) := (x,m) \prec (y,q) \vee (x = y \wedge m \sim^s q)$
    \item $(x,m) \prec_1 (y,q) := (x,m) \prec (y,q) \wedge \neg\exists r. (m <^s r \wedge r <^s q)$
\end{itemize}
We can now write the translation. Fix $\form$, formula of \lsptl. We construct $\translated{\form}(x,m)$ with two free variables $x$ and $m$,
such that for any sub-event $(e,f)$ of $P$, $P, [x \mapsto e, m \mapsto f] \models \translated{\form}(x,m)$ if and only if
$P, (e,f) \models \form$.
\begin{itemize}
    \item For any $a \in \Sigma$, $\translated{a}(x,m) = a(x)$
    \item $\translated{\starter}(x,m) = \neg\starter(x) \wedge x \sim^s m$
    \item $\translated{\terminator}(x,m) = \neg\terminator(x) \wedge \neg\exists y. (x \parallel y \wedge m <^s y)$
    \item $\translated{\neg\form}(x,m) = \neg\translated{\form}(x,m)$
    \item $\translated{\form_1 \wedge \form_2}(x,m) = \translated{\form_1}(x,m) \wedge \translated{\form_2}(x,m)$
    \item $\translated{\enext \form}(x,m) = \exists y.\exists q. (\fose(y,q) \wedge (x,m) \prec_1 (y,q) \wedge \translated{\form}(y,q))$
    \item $\translated{\ltlcirc \form}(x,m) = \exists y. (\fose(y,m) \wedge x \eventorder y \wedge \neg\exists z. x \eventorder z \wedge z \eventorder y \wedge \translated{\varphi}(y,m))$
    \item $\translated{\ltlrcirc \form}(x,m) = \exists y. (\fose(y,m) \wedge y \eventorder x \wedge \neg\exists z. y \eventorder z \wedge z \eventorder x \wedge \translated{\varphi}(y,m))$
    \item $\translated{\form_1 \until \form_2}(x,m) =
        \exists y. \exists q. [
            \fose(y,q) \wedge (x,m) \preccurlyeq (y,q) \wedge \translated{\form_2}(y,q) \wedge
            \forall z. \forall r. (\fose(z,r) \wedge (x,m) \preccurlyeq (z,r) \wedge (z,r) \prec (y,q)) \Rightarrow \translated{\form_1}(z,q)
        ]$
\end{itemize}
We can now write $\translated{\form} = \exists x. (\forall y. (\forall z. \neg z < y) \Rightarrow x \eventorder y) \wedge \translated{\form}(x,x)$, which
concludes the proof.

%% file: main.bbl
\begin{thebibliography}{10}

\bibitem{AmraneBCFZ24}
Amazigh Amrane, Hugo Bazille, Emily Clement, Uli Fahrenberg, and Krzysztof Ziemianski.
\newblock Presenting interval pomsets with interfaces.
\newblock In Uli Fahrenberg, Wesley Fussner, and Roland Gl{\"{u}}ck, editors, {\em Relational and Algebraic Methods in Computer Science - 21st International Conference, RAMiCS 2024, Prague, Czech Republic, August 19-22, 2024, Proceedings}, volume 14787 of {\em Lecture Notes in Computer Science}, pages 28--45. Springer, 2024.
\newblock \href {https://doi.org/10.1007/978-3-031-68279-7\_3} {\path{doi:10.1007/978-3-031-68279-7\_3}}.

\bibitem{AmraneBFF24}
Amazigh Amrane, Hugo Bazille, Uli Fahrenberg, and Marie Fortin.
\newblock Logic and languages of higher-dimensional automata.
\newblock In Joel~D. Day and Florin Manea, editors, {\em Developments in Language Theory - 28th International Conference, {DLT} 2024, G{\"{o}}ttingen, Germany, August 12-16, 2024, Proceedings}, volume 14791 of {\em Lecture Notes in Computer Science}, pages 51--67. Springer, 2024.
\newblock \href {https://doi.org/10.1007/978-3-031-66159-4\_5} {\path{doi:10.1007/978-3-031-66159-4\_5}}.

\bibitem{AmraneBFZ23}
Amazigh Amrane, Hugo Bazille, Uli Fahrenberg, and Krzysztof Ziemianski.
\newblock Closure and decision properties for higher-dimensional automata.
\newblock In Erika {\'{A}}brah{\'{a}}m, Clemens Dubslaff, and Silvia Lizeth~Tapia Tarifa, editors, {\em Theoretical Aspects of Computing - {ICTAC} 2023 - 20th International Colloquium, Lima, Peru, December 4-8, 2023, Proceedings}, volume 14446 of {\em Lecture Notes in Computer Science}, pages 295--312. Springer, 2023.
\newblock \href {https://doi.org/10.1007/978-3-031-47963-2_18} {\path{doi:10.1007/978-3-031-47963-2_18}}.

\bibitem{ClementEL25}
Emily Clement, Enzo Erlich, and J{\'{e}}r{\'{e}}my Ledent.
\newblock Expressivity of linear temporal logic for pomset languages of higher dimensional automata.
\newblock {\em CoRR}, abs/2410.12493, 2024.
\newblock URL: \url{https://doi.org/10.48550/arXiv.2410.12493}, \href {https://arxiv.org/abs/2410.12493} {\path{arXiv:2410.12493}}, \href {https://doi.org/10.48550/ARXIV.2410.12493} {\path{doi:10.48550/ARXIV.2410.12493}}.

\bibitem{DiekertG02completeLTL}
Volker Diekert and Paul Gastin.
\newblock {LTL} is expressively complete for mazurkiewicz traces.
\newblock {\em J. Comput. Syst. Sci.}, 64(2):396--418, 2002.
\newblock \href {https://doi.org/10.1006/jcss.2001.1817} {\path{doi:10.1006/jcss.2001.1817}}.

\bibitem{DiekertG06local-global}
Volker Diekert and Paul Gastin.
\newblock From local to global temporal logics over mazurkiewicz traces.
\newblock {\em Theoretical Computer Science}, 356(1):126--135, 2006.
\newblock In honour of Professor Christian Choffrut on the occasion of his 60th birthday.
\newblock URL: \url{https://www.sciencedirect.com/science/article/pii/S0304397506001095}, \href {https://doi.org/10.1016/j.tcs.2006.01.035} {\path{doi:10.1016/j.tcs.2006.01.035}}.

\bibitem{DiekertG06local}
Volker Diekert and Paul Gastin.
\newblock Pure future local temporal logics are expressively complete for mazurkiewicz traces.
\newblock {\em Inf. Comput.}, 204(11):1597--1619, 2006.
\newblock \href {https://doi.org/10.1016/j.ic.2006.07.002} {\path{doi:10.1016/j.ic.2006.07.002}}.

\bibitem{DiekertG08}
Volker Diekert and Paul Gastin.
\newblock First-order definable languages.
\newblock In J{\"{o}}rg Flum, Erich Gr{\"{a}}del, and Thomas Wilke, editors, {\em Logic and Automata: History and Perspectives [in Honor of Wolfgang Thomas]}, volume~2 of {\em Texts in Logic and Games}, pages 261--306. Amsterdam University Press, 2008.

\bibitem{Dijkstra68}
Edsger~W. Dijkstra.
\newblock The structure of "the"-multiprogramming system.
\newblock {\em Commun. {ACM}}, 11(5):341--346, 1968.
\newblock \href {https://doi.org/10.1145/363095.363143} {\path{doi:10.1145/363095.363143}}.

\bibitem{FahrenbergJSZ21}
Uli Fahrenberg, Christian Johansen, Georg Struth, and Krzysztof Ziemianski.
\newblock Languages of higher-dimensional automata.
\newblock {\em Math. Struct. Comput. Sci.}, 31(5):575--613, 2021.
\newblock \href {https://doi.org/10.1017/S0960129521000293} {\path{doi:10.1017/S0960129521000293}}.

\bibitem{FahrenbergJSZ22}
Uli Fahrenberg, Christian Johansen, Georg Struth, and Krzysztof Ziemianski.
\newblock A kleene theorem for higher-dimensional automata.
\newblock In Bartek Klin, Slawomir Lasota, and Anca Muscholl, editors, {\em 33rd International Conference on Concurrency Theory, {CONCUR} 2022, September 12-16, 2022, Warsaw, Poland}, volume 243 of {\em LIPIcs}, pages 29:1--29:18. Schloss Dagstuhl - Leibniz-Zentrum f{\"{u}}r Informatik, 2022.
\newblock \href {https://doi.org/10.4230/LIPIcs.CONCUR.2022.29} {\path{doi:10.4230/LIPIcs.CONCUR.2022.29}}.

\bibitem{FahrenbergZ23}
Uli Fahrenberg and Krzysztof Ziemianski.
\newblock A myhill-nerode theorem for higher-dimensional automata.
\newblock In Lu{\'{i}}s Gomes and Robert Lorenz, editors, {\em Application and Theory of Petri Nets and Concurrency - 44th International Conference, {PETRI} {NETS} 2023, Lisbon, Portugal, June 25-30, 2023, Proceedings}, volume 13929 of {\em Lecture Notes in Computer Science}, pages 167--188. Springer, 2023.
\newblock \href {https://doi.org/10.1007/978-3-031-33620-1_9} {\path{doi:10.1007/978-3-031-33620-1_9}}.

\bibitem{GastinK07PSPACE}
Paul Gastin and Dietrich Kuske.
\newblock Uniform satisfiability in {PSPACE} for local temporal logics over mazurkiewicz traces.
\newblock {\em Fundam. Informaticae}, 80(1-3):169--197, 2007.
\newblock URL: \url{https://journals.sagepub.com/doi/abs/10.3233/FUN-2007-801-310}, \href {https://doi.org/10.3233/FUN-2007-801-310} {\path{doi:10.3233/FUN-2007-801-310}}.

\bibitem{GoubaultM12}
Eric Goubault and Samuel Mimram.
\newblock Formal relationships between geometrical and classical models for concurrency.
\newblock In Lisbeth Fajstrup, Eric Goubault, and Martin Raussen, editors, {\em Proceedings of the workshop on Geometric and Topological Methods in Computer Science, {GETCO} 2010, Aalborg, Denmark, January 11-15, 2010}, volume 283 of {\em Electronic Notes in Theoretical Computer Science}, pages 77--109. Elsevier, 2010.
\newblock \href {https://doi.org/10.1016/j.entcs.2012.05.007} {\path{doi:10.1016/j.entcs.2012.05.007}}.

\bibitem{JoyalNW96}
Andr{\'{e}} Joyal, Mogens Nielsen, and Glynn Winskel.
\newblock Bisimulation from open maps.
\newblock {\em Inf. Comput.}, 127(2):164--185, 1996.
\newblock \href {https://doi.org/10.1006/inco.1996.0057} {\path{doi:10.1006/inco.1996.0057}}.

\bibitem{Kamp68}
Johan Anthony~Willem Kamp.
\newblock {\em Tense Logic and the Theory of Linear Order}.
\newblock {PhD} thesis, Computer Science Department, University of California at Los~Angeles, USA, 1968.

\bibitem{Lamport86}
Leslie Lamport.
\newblock On interprocess communication. part {I:} basic formalism.
\newblock {\em Distributed Comput.}, 1(2):77--85, 1986.
\newblock \href {https://doi.org/10.1007/BF01786227} {\path{doi:10.1007/BF01786227}}.

\bibitem{Mazurkiewicz86Trace}
Antoni~W. Mazurkiewicz.
\newblock Trace theory.
\newblock In Wilfried Brauer, Wolfgang Reisig, and Grzegorz Rozenberg, editors, {\em Petri Nets: Central Models and Their Properties, Advances in Petri Nets 1986, Part II, Proceedings of an Advanced Course, Bad Honnef, Germany, 8-19 September 1986}, volume 255 of {\em Lecture Notes in Computer Science}, pages 279--324. Springer, 1986.
\newblock \href {https://doi.org/10.1007/3-540-17906-2_30} {\path{doi:10.1007/3-540-17906-2_30}}.

\bibitem{McNaughton71}
Robert McNaughton and Seymour Papert.
\newblock {\em Counter-Free Automata}, volume Research Monograph, 65.
\newblock The M.I.T. Press, Cambridge, Mass., 1971.

\bibitem{Pratt91}
Vaughan~R. Pratt.
\newblock Modeling concurrency with geometry.
\newblock In David~S. Wise, editor, {\em Conference Record of the Eighteenth Annual {ACM} Symposium on Principles of Programming Languages, Orlando, Florida, USA, January 21-23, 1991}, pages 311--322. {ACM} Press, 1991.
\newblock \href {https://doi.org/10.1145/99583.99625} {\path{doi:10.1145/99583.99625}}.

\bibitem{Prisacariu10}
Cristian Prisacariu.
\newblock Modal logic over higher dimensional automata.
\newblock In Paul Gastin and Fran{\c{c}}ois Laroussinie, editors, {\em {CONCUR} 2010 - Concurrency Theory, 21th International Conference, {CONCUR} 2010, Paris, France, August 31-September 3, 2010. Proceedings}, volume 6269 of {\em Lecture Notes in Computer Science}, pages 494--508. Springer, 2010.
\newblock \href {https://doi.org/10.1007/978-3-642-15375-4_34} {\path{doi:10.1007/978-3-642-15375-4_34}}.

\bibitem{Rabinovich14}
Alexander Rabinovich.
\newblock A proof of kamp's theorem.
\newblock {\em Logical Methods in Computer Science}, Volume 10, Issue 1:730, 2014.
\newblock URL: \url{https://lmcs.episciences.org/730}, \href {https://doi.org/10.2168/LMCS-10(1:14)2014} {\path{doi:10.2168/LMCS-10(1:14)2014}}.

\bibitem{Schutzenberger65}
Marcel~Paul Sch{\"u}tzenberger.
\newblock On finite monoids having only trivial subgroups.
\newblock {\em Information and Control}, 8(2):190--194, 1965.
\newblock URL: \url{https://www.sciencedirect.com/science/article/pii/S0019995865901087}, \href {https://doi.org/10.1016/S0019-9958(65)90108-7} {\path{doi:10.1016/S0019-9958(65)90108-7}}.

\bibitem{Thiagarajan94traceLTL}
P.~S. Thiagarajan.
\newblock A trace based extension of linear time temporal logic.
\newblock In {\em Proceedings of the Ninth Annual Symposium on Logic in Computer Science {(LICS} '94), Paris, France, July 4-7, 1994}, pages 438--447. {IEEE} Computer Society, 1994.
\newblock \href {https://doi.org/10.1109/LICS.1994.316047} {\path{doi:10.1109/LICS.1994.316047}}.

\bibitem{ThiagarajanW02completeLTL}
P.S. Thiagarajan and I.~Walukiewicz.
\newblock An expressively complete linear time temporal logic for mazurkiewicz traces.
\newblock {\em Information and Computation}, 179(2):230--249, 2002.
\newblock URL: \url{https://www.sciencedirect.com/science/article/pii/S0890540101929566}, \href {https://doi.org/10.1006/inco.2001.2956} {\path{doi:10.1006/inco.2001.2956}}.

\bibitem{Glabbeek06}
Rob~J. van Glabbeek.
\newblock On the expressiveness of higher dimensional automata.
\newblock {\em Theor. Comput. Sci.}, 356(3):265--290, 2006.
\newblock \href {https://doi.org/10.1016/j.tcs.2006.02.012} {\path{doi:10.1016/j.tcs.2006.02.012}}.

\bibitem{ZouariZF24}
Safa Zouari, Krzysztof Ziemianski, and Uli Fahrenberg.
\newblock Bisimulations and logics for higher-dimensional automata.
\newblock {\em CoRR}, abs/2402.01589, 2024.
\newblock \href {https://arxiv.org/abs/2402.01589} {\path{arXiv:2402.01589}}, \href {https://doi.org/10.48550/arXiv.2402.01589} {\path{doi:10.48550/arXiv.2402.01589}}.

\end{thebibliography}
